\documentclass[sts]{imsart}

\RequirePackage[OT1]{fontenc}
\usepackage{amsthm,amsmath,amssymb,amsfonts}
\usepackage{bbm,bm}	
\usepackage{natbib}
\usepackage{enumitem}
\usepackage{listings, caption, graphicx}
\usepackage{xcolor,colortbl,longtable,multirow}
\usepackage{rotating}
\usepackage{tabularx,booktabs}
\usepackage{tikz,tcolorbox}
\usepackage{hyperref}

\startlocaldefs

\numberwithin{equation}{section}
\theoremstyle{plain}
\newtheorem{theorem}{Theorem}[section]
\newtheorem{proposition}[theorem]{Proposition}
\newtheorem{corollary}[theorem]{Corollary}

\theoremstyle{definition}

\theoremstyle{remark}
\newtheorem*{remark}{Remark}

\newcounter{nalg}[section] 
\renewcommand{\thenalg}{\thesection .\arabic{nalg}} 
\DeclareCaptionLabelFormat{algocaption}{Algorithm \thenalg} 
\lstnewenvironment{algorithm}[1][] 
{   
	\refstepcounter{nalg} 
	\captionsetup{labelformat=algocaption,labelsep=colon} 
	\lstset{ 
                tabsize=3,
                backgroundcolor=\color{gray!10},
		mathescape=true,
		frame=tB,
		columns=space-flexible,
		numbers=left, 
		numberstyle=\tiny,
		basicstyle=\small,
		keywordstyle=\color{black}\bfseries\em,
		keywords={,input, output, return, if, else, foreach, while, begin, end, simulate, set, update, break, for,} 
		numbers=left,
		xleftmargin=.01\textwidth,
		#1 
	}
}
{}

\newcommand{\EE}{\mathbb{E}}
\newcommand{\PP}{\mathbb{P}}
\newcommand{\NN}{\mathbb{N}}
\newcommand{\RR}{\mathbb{R}}
\newcommand{\dd}{\mathrm{d}}
\newcommand{\Var}{\mathrm{Var}}
\newcommand{\Cov}{\mathrm{Cov}}
\newcommand{\eps}{\varepsilon}
\newcommand{\Ex}{\text{Ex}}
\newcommand{\spec}{v}

\newcommand{\evi}{\xi}
\newcommand{\vario}{\gamma}

\endlocaldefs

\begin{document}

\begin{frontmatter}
	
\title{A comparative tour\\ through the simulation algorithms for max-stable processes}

\runtitle{Simulation algorithms for max-stable processes}

\begin{aug}
\author{\fnms{Marco} \snm{Oesting}\ead[label=e1]{marco.oesting@mathematik.uni-stuttgart.de}}
\and
\author{\fnms{Kirstin} \snm{Strokorb}\ead[label=e2]{strokorbk@cardiff.ac.uk}}

\runauthor{M.~Oesting and K.~Strokorb}

\affiliation{University of Stuttgart and Cardiff University}

\address{\small Marco Oesting, University of Stuttgart, Stuttgart Center for Simulation Science \& Institute for Stochastics and Applications, 70569 Stuttgart, Germany, \printead{e1}.}

\address{\small Kirstin Strokorb, Cardiff University, School of Mathematics, Cardiff CF24 4AG, UK, \printead{e2}.}
\end{aug}

\begin{abstract}
Being the max-analogue of $\alpha$-stable stochastic processes, max-stable processes form one of the fundamental classes of stochastic processes. With the arrival of sufficient computational capabilities, they have become a benchmark in the analysis of spatio-temporal extreme events. Simulation is often a necessary part of inference of certain characteristics, in particular for future spatial risk assessment. In this article we give an overview over existing procedures for this task, put them into perspective of one another and use some new theoretical results to make comparisons with respect to their properties.\vspace*{1mm}
\end{abstract}



\end{frontmatter}

\section{Introduction} \label{sec:intro}

The severe consequences of extreme events such as strong windstorms, heavy precipitation or heat waves emphasize the need of appropriate statistical models for these types of events. To adequately assess the associated risk, not only the intensity, but also the spatial or spatio-temporal extent of extremes has to be taken into account. Motivated by the central limit theorem, classical geostatistics typically applies Gaussian processes to model the bulk of the distribution and the dependence structure of continuous variables.
Being the natural analogue to Gaussian or, more generally, $\alpha$-stable processes for maxima, max-stable processes are frequently used to model spatial and spatio-temporal extremes, specifically in environmental applications \citep[cf.][and references therein]{davison-etal-12, davison-etal-18}.

The prevalence of max-stable processes as a benchmark is justified by the fact that they arise as the only possible location-scale max-limits of stochastic processes in the following way: Let $X_1,X_2,X_3,\dots$ be independent copies of a real-valued 
process $X=\{X({\bm x})\}_{{\bm x} \in S}$ on some {locally compact metric space} $S$. If there exist 
suitable location-scale norming sequences $a_n=a_n({\bm x}) > 0$ and 
$b_n=b_n({\bm x}) \in \RR$, such that the law of
\begin{align*}
\bigg\{ \max_{i=1,\dots,n} \frac{X_i({\bm x})-b_n({\bm x})}{a_n({\bm x})} \bigg\}_{{\bm x} \in S}
\end{align*}
converges in distribution to a stochastic process $Z=\{Z({\bm x})\}_{{\bm x} \in S}$,
then the resulting limit process $Z$ necessarily satisfies a stability property
with respect to the maximum operation. More precisely, 
\begin{align*}
\bigg\{ \max_{i=1,\dots,n} \frac{Z_i({\bm x})-d_n({\bm x})}{c_n({\bm x})}  \bigg\}_{{\bm x} \in S} = \{Z({\bm x})\}_{{\bm x} \in S} \quad \text{in distribution}
\end{align*}
for independent copies $Z_1,Z_2,Z_3,\dots$ of $Z$ and appropriate location-scale norming 
sequences $c_n=c_n({\bm x}) > 0$ and $d_n=d_n({\bm x}) \in \RR$. In this sense,
the process $Z$ is \emph{max-stable} and the process $X$ lies in its 
\emph{max-domain of attraction}. 

Over the last decades, max-stable processes have gained attention from several research perspectives.
In the 80s and early 90s, they have mainly been studied from a 
probabilistic angle, resulting, for instance, in a full characterization of the
class of sample-continuous simple max-stable processes, see \citet{dh84}, 
\citet{norberg86} and \citet{ghv90}, among others. This work has been 
complemented by a precise description of the corresponding max-domain of 
attraction in \citet{dhl01}. Since the early 2000s, methods for statistical
inference have been developed and, in parallel, suitable models for subclasses
of max-stable processes have been introduced. Important examples for such 
models comprise Gaussian extreme value processes \citep{smith90}, extremal 
Gaussian processes \citep{schlather02}, Brown-Resnick processes \citep{ksdh09}
and extremal-$t$ processes \citep{opitz13} providing a generalization of 
extremal Gaussian processes. With these flexible models and tools at hand, 
max-stable processes have become attractive for practitioners from various 
areas, in particular from environmental sciences.

However, a serious drawback is that most probabilistic properties of max-stable
processes are analytically intractable. Therefore, simulation is often a necessary part of inference of certain characteristics, in particular for future spatial risk assessment. Meanwhile, starting from a general idea coined by \citet{schlather02}, a number of approaches to the simulation of max-stable processes have emerged:  They include both approximate \citep{oks12,os18} and exact \citep{dm15,deo16,osz18,lbdm16} simulation procedures, some of them focusing on the particularly difficult problem of simulating within the subclass of Brown-Resnick processes. A first overview over some of these methods has been given in \citet{ord16}. The present article extends and updates this overview. New theoretical results allow to put the different methods into perspective of one another and to make comparisons with respect to their theoretical properties and their performance in numerical experiments.

Our text is structured as follows. To illustrate the need for efficient and accurate simulation of max-stable processes in an application, Section~\ref{sec:dataexample} describes a worked example addressing questions that involve the (joint) distribution of areal maxima over certain time horizons. Subsequently, we review existing simulation approaches in Section~\ref{sec:survey}. Besides generic algorithms for the 
simulation of arbitrary  max-stable processes, we also present more specific 
procedures that have been developed for the important and popular subclass of
Brown-Resnick processes (Section~\ref{sec:BR}). In Section~\ref{sec:eff-acc}, 
we provide new theoretical results that allow us to evaluate and compare the
simulation approaches with respect to their efficiency and their accuracy. 
The results of a numerical study comparing the performance of the generic algorithms in a wide range of scenarios are reported in Section~\ref{sec:study}.
We conclude with a discussion in Section~\ref{sec:discussion} including further practical advice. 

Full reproducible code of the analysis of Section~\ref{sec:dataexample} and the numerical study of Section~\ref{sec:study}, all of which were carried out in the statistical software environment \texttt{R} \citep{gnuR}, are available at   \href{https://github.com/Oesting/Comparative-tour-through-simulation-algorithms-for-max-stable-processes}{\texttt{https://github.com/Oesting/Comparative-tour-}}
\href{https://github.com/Oesting/Comparative-tour-through-simulation-algorithms-for-max-stable-processes}{\texttt{through-simulation-algorithms-for-max-stable-processes}}.\vspace*{3mm}

\paragraph{Marginal standardization.} For convenience, let us recall a structural result at the very start. We shall assume throughout that all marginal distributions of a max-stable process $Z=\{Z({\bm x})\}_{{\bm x} \in S}$ are non-degenerate (i.e.\ not concentrated on a single value). This implies that the marginal distributions of $Z$ are necessarily \emph{Generalized Extreme Value (GEV) distributions} (\cite{jenkinson55} and \cite{vonMises36})
\begin{align*}
\PP(Z(\bm x) \leq z) = G_{\evi(\bm x)} \bigg(\frac{z-\mu(\bm x)}{\sigma(\bm x)}\bigg), \qquad G_{\evi(\bm x)}(z)=\exp(-(1+\evi(\bm x) z)^{-1/\evi(\bm x)}_+),
\end{align*}
with shape, location and scale parameters $\evi(\bm x),\mu(\bm x) \in \RR$ and $\sigma(\bm x) > 0$, a result that goes back to \cite{frechet27} and \cite{fishertippett28} and was first rigorously proved in \cite{gnedenko43}.
As the max-stability property is preserved under marginal transformations 
within the class of GEV distributions, attention is often drawn to 
max-stable processes with fixed marginal distributions such as the class of 
so-called \emph{simple max-stable processes}, the ones that have standard unit Fr\'echet 
marginal distributions, i.e.\ $\PP(Z(\bm x) \leq z) = \exp(-1/z)$ for all $z>0$
and $\bm x \in S$. To make this precise, if $Z$ is a general max-stable process with shape function $\evi$, location function $\mu$ and scale function $\sigma$, then
\begin{align*}
  \bigg\{
  \left(
  1+ \evi(\bm x) \frac{Z(\bm x)-\mu(\bm x)}{\sigma(\bm x)}
 \right)^{1/\evi(\bm x)}
  \bigg\}_{\bm{x} \in S}
\end{align*}
is a simple max-stable process and, conversely, if $Z$ is a simple max-stable process, then
\begin{align*}
\bigg\{\sigma(\bm x) \frac{Z(\bm{x})^{\evi(\bm x)} -1 }{\evi(\bm x)}  + \mu(\bm x) \bigg\}_{\bm{x} \in S}
\end{align*}
is a max-stable process with general GEV marginal distributions.
Thus, the simple max-stable process associated with a general max-stable process in this way is a natural object to encapsulate its dependence structure (similarly to the disentanglement of a multivariate distribution into marginal distributions and a copula). In addition, if we are concerned about the simulation of a max-stable process with given marginal distributions, this amounts to the simulation of a simple max-stable process and transforming it thereafter to the desired marginal distributions.

\section{Data example}\label{sec:dataexample}

In statistical practice, max-stable processes are used to approximate pointwise block maxima taken over sufficiently long time periods. In this data example we consider daily maximum summer temperatures from 1990 to 2019 that were measured at 18 inland stations in the Netherlands and are freely available from \url{http://projects.knmi.nl/klimatologie/daggegevens/selectie.cgi} (cf.~Figure~\ref{fig:OneStation})
to answer questions about \dots
\begin{enumerate}[label={(Q\arabic*)},itemsep=0mm]
\item \ldots the distribution of the maximum inland temperature over a period of two weeks in summer,
\item $\ldots$ the probability that three subregions (SW, SE, NE, cf.~Figure \ref{fig:Q2}) experience a joint exceedance of $38^\circ$C within the same period of 14 days in summer.
\end{enumerate}

To be more precise, when we say ``inland'', we mean every location in the Netherlands that is 15\,km away from the coast\footnote{Coastline data were obtained from \texttt{rnaturalearth} \citep{R_rnaturalearth} and the package \texttt{geosphere} \citep{R_geosphere} used to compute coast distances. The data at stations in this region exhibit a more homogeneous behaviour compared to including stations closer to the coast}. In this study it is represented by a grid $S=\{{\bm y_j}\}_{j=1}^{4712}$ of 4712 inland locations with (approximate) grid distance 2.5\,km as displayed in Figures~\ref{fig:OneStation} (left) and \ref{fig:LocationFunctionECplot} (right).  As regards the temperature data, exemplarily, Figure~\ref{fig:OneStation} shows the daily maximum temperatures on 84 summer days (5 June to 27 August) from 1990 to 2019 as well as their 14 day block maxima at one of the 18 stations, Deelen. Considering temporal dependence, we did not detect strong evidence against independence among 14 day block maxima, whilst spatial dependence seems to be of paramount importance, especially in the extreme values (as we shall see confirmed below). Hence, our working assumption is that these 14 day block maxima can be considered as 180 independent samples\footnote{Missing values only occur at two sites, Stavoren and Arcen, where we lack knowledge about the first 29 summer days in 1990. The corresponding 14 day block maxima have been removed from the analysis.} (30 years with 6 blocks/year) of a max-stable process $\{Z({\bm x})\}_{{\bm x} \in S}$ on the inland part $S$ of the Netherlands measured at 18 sites ${\bm x_1}, {\bm x_2}, {\bm x_3}, \dots, {\bm x_{18}} \in S$ as displayed in Figures~\ref{fig:OneStation} (left) and \ref{fig:LocationFunctionECplot} (right). \vspace*{3mm}

\begin{figure}
  \centering
  \includegraphics[width=\linewidth]{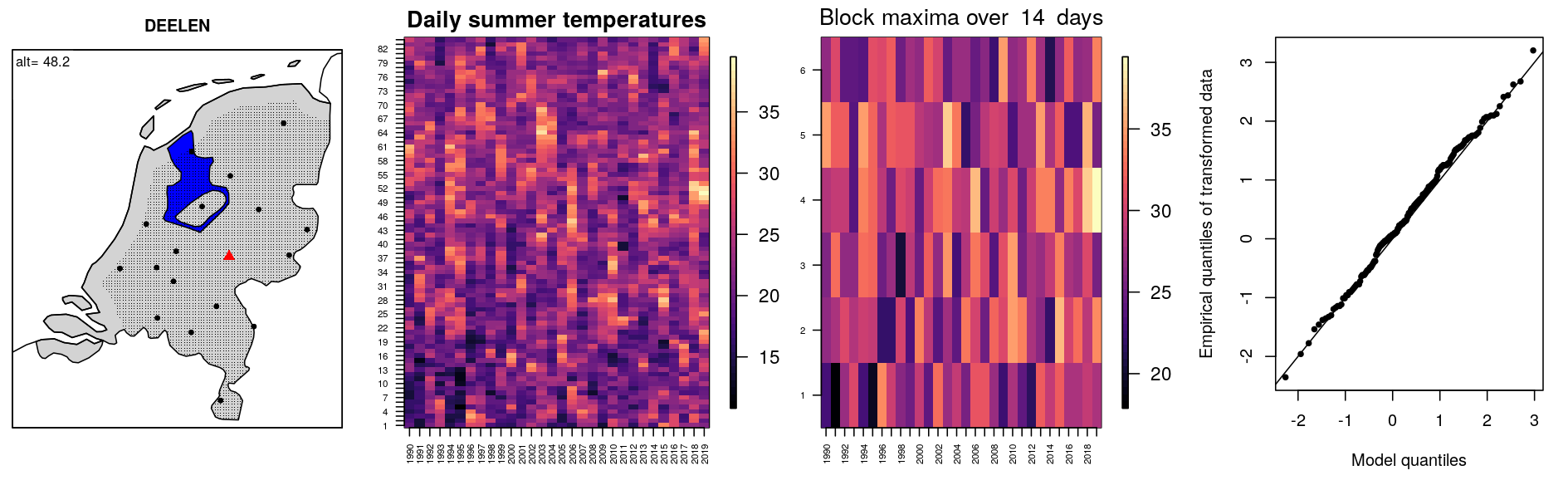}
  \caption{\small Daily maximum summer temperatures on 84 summer days (5 June - 27 August) from 1990 to 2019 at Deelen and their 14 day block maxima. Their location-scale transformation and the model quantiles arise from the marginal estimates as described in Section~\ref{sec:dataexample}.}
  \label{fig:OneStation}
\end{figure}

\paragraph{Estimation.} Following the paradigm of separating estimation of the marginal and dependence structure, a standard approach in Extreme Value Analysis, we first fitted a marginal GEV distribution $G(z;\mu,\sigma,\evi)$ to the 14 day block maxima via maximization of the independence likelihood. Here, while assuming the shape parameter $\xi$ and the scale parameter $\sigma$ to be constant, geographical information on the measurement stations serve as covariates in the location parameter $\mu = \mu_0 + \mu_1 \, \text{\texttt{longitude}} + \mu_2 \, \text{\texttt{latitude}} + \mu_3 \, \text{\texttt{altitude}}$. This maximum likelihood estimator for the GEV parameters is consistent and asymptotically normal for $\xi>-0.5$ (cf.~\cite{chandler2007inference} and \cite{axeljohan17}) and is, for instance, readily available in the \texttt{R} package \texttt{extRemes} \citep{R_extRemes}. In our case, the estimated GEV distributions are of Weibull-type ($\widehat \evi \approx -0.27 <0$). Exemplarily, Figure~\ref{fig:OneStation} (right) displays a QQ-plot of the site-specific location-scale-standardized block maxima against the model quantiles of $G(x;\widehat \evi)$. The marginal GEV-fit has been found reasonable at almost all sites. Figure~\ref{fig:LocationFunctionECplot} (left) displays the estimated location function $\widehat \mu$, when evaluated on the inland grid. Altitude information on the inland grid stems from the ASTER Global Digital Elevation Model V2 \citep{ASTERv2} and was accessed with \texttt{geonames} \citep{R_geonames}.

\begin{figure}
  \centering
  \mbox{
    \includegraphics[height=5cm]{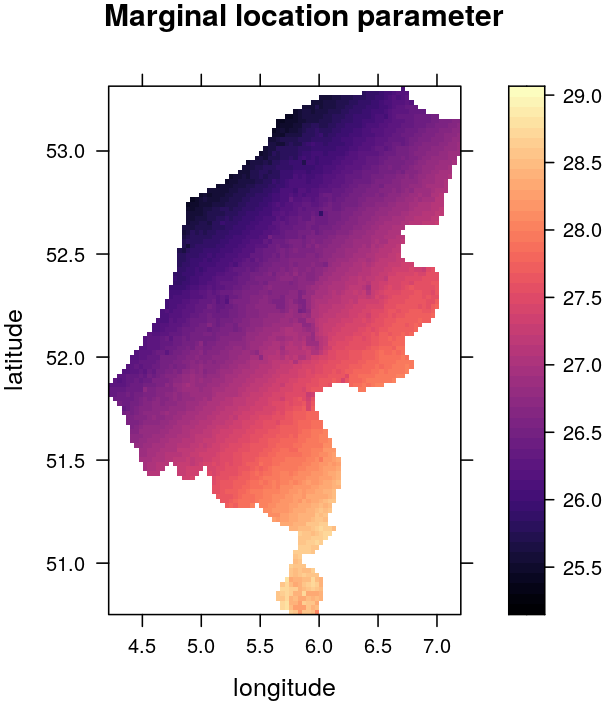}
    \hspace{5mm}
  \includegraphics[height=5cm]{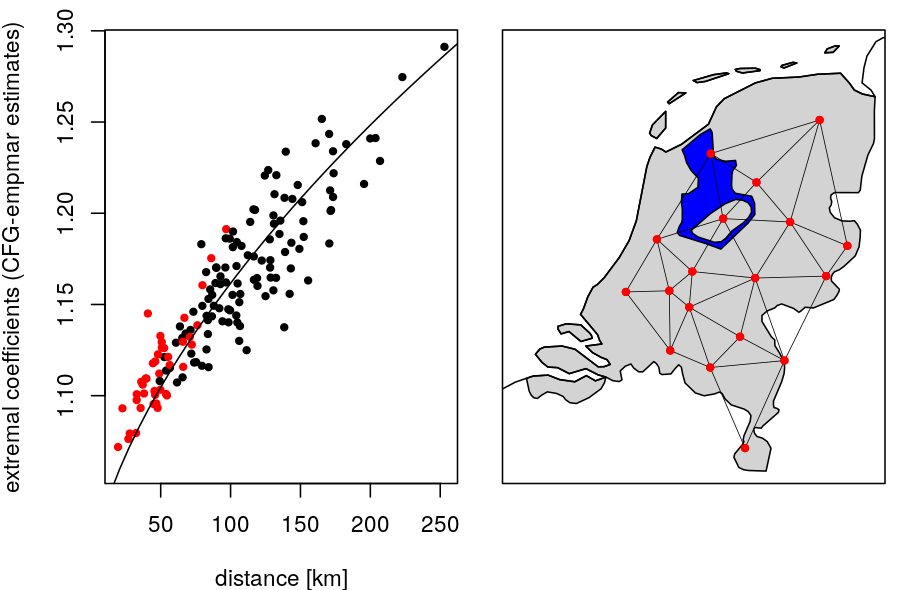}
  }
  \caption{\small Left: Location function of the marginal GEV fit evaluated on the entire inland grid. Middle: Non-parametric estimates of bivariate extremal coefficients corresponding to station pairs as estimated by the non-parametric procedure of \cite{cfg97} (dots) and theoretical counterpart of the fitted isotropic spatial max-stable model (line). Red dots indicate pairs that have been included in the procedure of \cite{ekks16} to estimate the spatial model. Right: Such pairs are connected in this graph.}
  \label{fig:LocationFunctionECplot}
\end{figure}

Subsequently, what remains to be modeled, is the dependence structure of $Z$, or, equivalently, the associated simple max-stable process after marginal standardization as described in Section~\ref{sec:intro}. For simplicity, we decided for an isotropic Brown-Resnick process with variogram $\vario(\bm{h})=\lVert \bm{h}/s \rVert^\alpha$. Details about this process can be found in Section~\ref{sec:BR} here. In addition to the estimation of the six marginal parameters, this amounts to the estimation of two further parameters to account for the spatial dependence, a smoothness parameter $\alpha \in (0,2)$ and a spatial scale $s>0$. An approach for this task, for which consistency and asymptotic normality have been established even under max-domain of attraction conditions, is the M-estimator of spatial tail dependence \citep{ekks16}. We chose it here as a generic approach with good large sample properties that is not specifically tailored to the class of Brown-Resnick processes. Alternatively, one might consider likelihood based methods \citep{hdg16lik} or more bespoke methods such as the Peaks-over-threshold-inspired estimation of \cite{emks15} that was proposed specifically for Brown-Resnick processes. The M-estimator relies on a selection of bivariate distributions only, i.e.\ pairs among the measurement stations, and a number $k$ of upper order statistics to be taken into account for the estimation. Our choice of pairs is displayed in Figure~\ref{fig:LocationFunctionECplot} (right) and we used $k=50$. In the related package \texttt{tailDepFun} \citep{R_tailDepFun}, we used the option \texttt{iterate=T} to improve and update the (internal) distance weight matrix.

A typical sanity check after fitting a spatial max-stable model to station data is a comparison of pairwise non-parametric estimates of bivariate extremal coefficients $\theta({\bm x}_i,{\bm x}_j)$ (cf.~\eqref{eq:extremal-coefficient} for $K=\{{\bm x}_i,{\bm x}_j\}$ in Section~\ref{sec:eff}) with the theoretical extremal coefficient function $\theta(\lVert {\bm x} - {\bm y}\rVert)$  of the estimated spatial model as displayed in Figure~\ref{fig:LocationFunctionECplot} (middle). Our bivariate non-parametric estimates are based on the procedure of \cite{cfg97} as implemented in \texttt{evd} \citep{R_evd}. The very low extremal coefficients, even for large distances, hint already at a strong extremal dependence among high temperatures.\vspace*{3mm}

\paragraph{Simulation.} Finally, we draw 60\,000 independent samples $Z_i = \{Z_i({\bm y_j})\}_{j=1}^{4712}$  from the fitted max-stable model $Z_i {\sim} Z$ ($i=1,\dots,60\,000$) on the inland grid $({\bm y_j})_{j=1}^{4712}$, which can be interpreted in this context as 60\,000 (grid-point-wise) temperature maxima of a 14 day summer interval. All simulations were carried out using the exact extremal functions methodology from \cite{deo16} (cf.~Section~\ref{sec:EF} below). Normal random variables therein were generated with the fast pseudo random number generator from \texttt{dqrng} \citep{R_dqrng}. Figure~\ref{fig:SixSimulationsQ1} shows the first six simulations and a histogram of the corresponding 60\,000 inland maxima, the latter being our data-driven answer to (Q1).

Since each summer consists of six such blocks in our setup, one may see this as sampling 10\,000 times from a summer that is represented by the data. In this sense the $(1-1/(6\ell))$-empirical quantile constitutes a return level estimate for $\ell$ years. Figure~\ref{fig:SixSimulationsQ1} (right) contains these estimates for $\ell\in\{10,100,1000\}$. Since our estimated model has Weibull margins, the maximum upper endpoint across the grid is finite and gives us an estimate for the maximum summer temperature across the entire inland. According to this model it would be $43.1^\circ$C.

\begin{figure}
  \centering
  \mbox{
    \includegraphics[height=5cm]{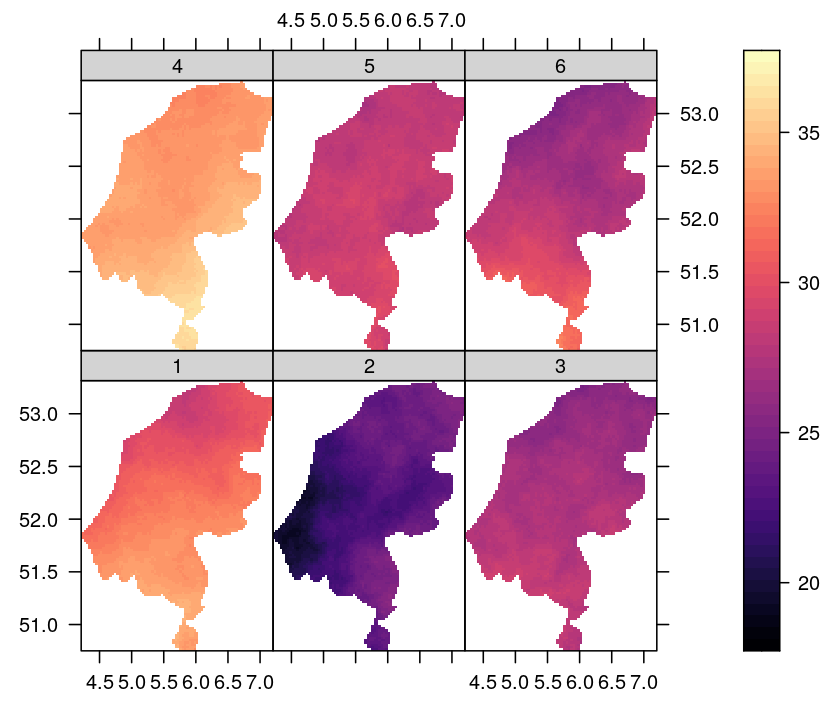}
    \hspace{5mm}
        \includegraphics[height=5cm]{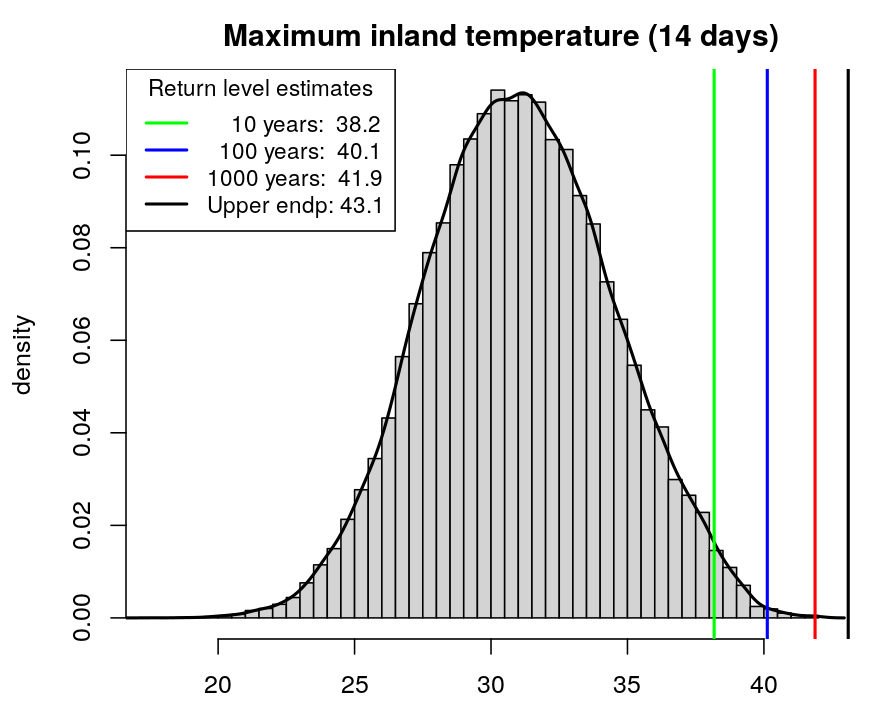}
  }
  \caption{\small Left: The first six out of 60\,000 simulations from the estimated max-stable process on the inland grid after transformation to the estimated GEV marginal distributions. Right: Histogram of the corresponding 60\,000 inland maxima. Return levels need to be interpreted with caution (cf.~Section~\ref{sec:dataexample}).} 
  \label{fig:SixSimulationsQ1}
\end{figure}

With regard to (Q2) we would like to draw attention to Figure~\ref{fig:Q2}, which shows a scatterplot of the 60\,000 joint areal maxima
\begin{align*}
  \Big\{\big(\max_{y_j \in \text{SW}}Z_i({\bm y_j}),\max_{y_j \in \text{SE}}Z_i({\bm y_j}),\max_{y_j \in \text{NE}}Z_i({\bm y_j})\big)\Big\}_{i=1}^{60\,000}
  \end{align*}
across the three marked regions. In particular the plot depicts that in our model dependence among high temperatures increases. In other words, particularly high temperatures are more likely to be experienced across a wider range of the inland. So our purely data-driven model would be in line with a (physical) theory that supports a single cause in such cases, a theory in which particularly large values may arise only during a wide-spread heatwave. Counting the joint exceedances of $38^\circ$C in all three regions (marked red in Figure~\ref{fig:Q2}), leads to an estimate of $\widehat p = 0.0046$ for a joint exceedance during a 14 day summer interval, our data-driven answer to Q2. Referring to an entire summer, one can interpret this as an estimate of $1-(1- \widehat p)^6 \approx 6 \widehat p \approx 2.75\%$ for the probability of a summer, in which $38^\circ$C is hit simultaneously in all three regions during at least one of the six 14 day summer periods. \vspace*{3mm}

\begin{figure}
  \centering
  \mbox{
    \includegraphics[height=5cm]{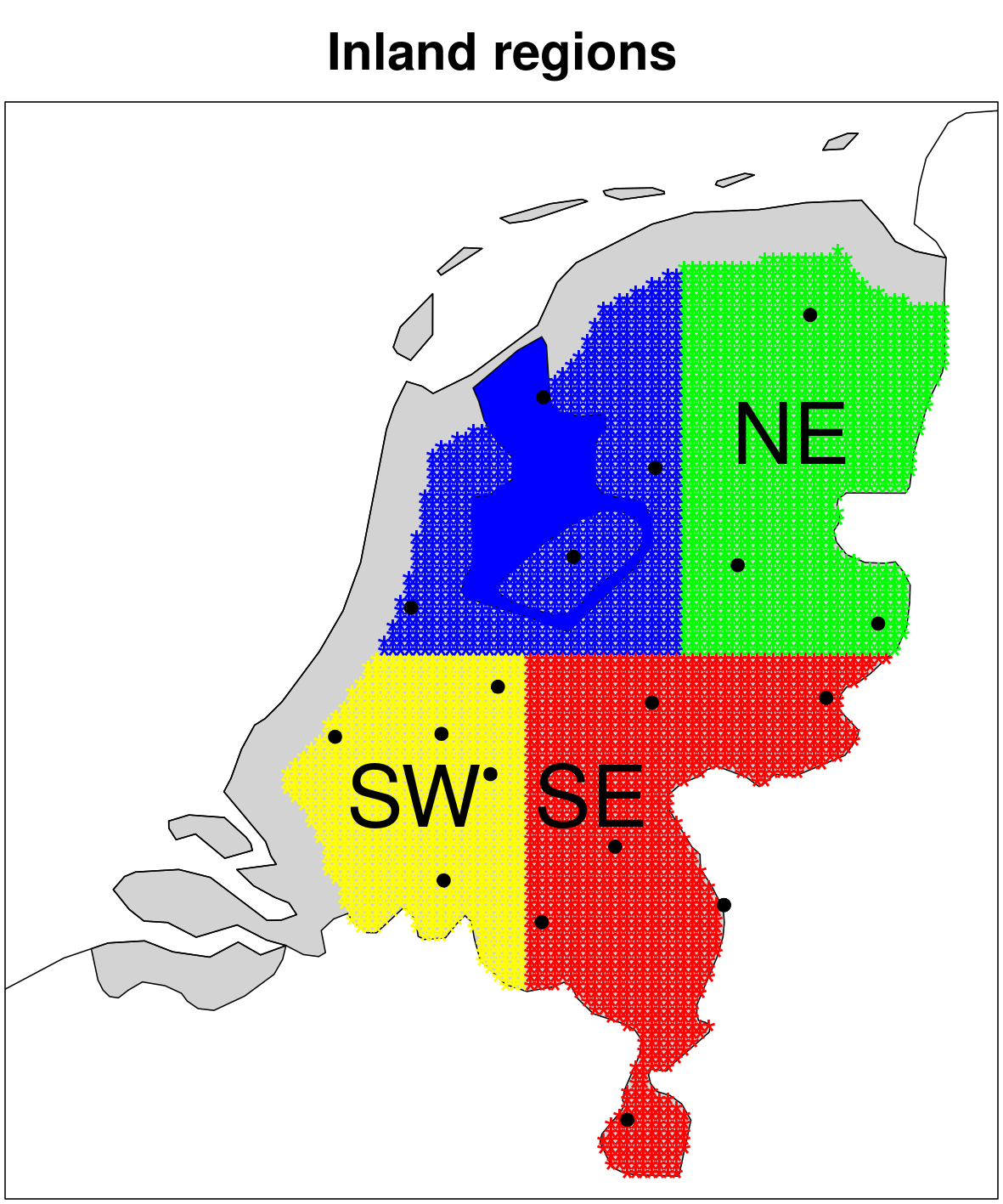}
    \hspace{5mm}
    \includegraphics[height=5cm]{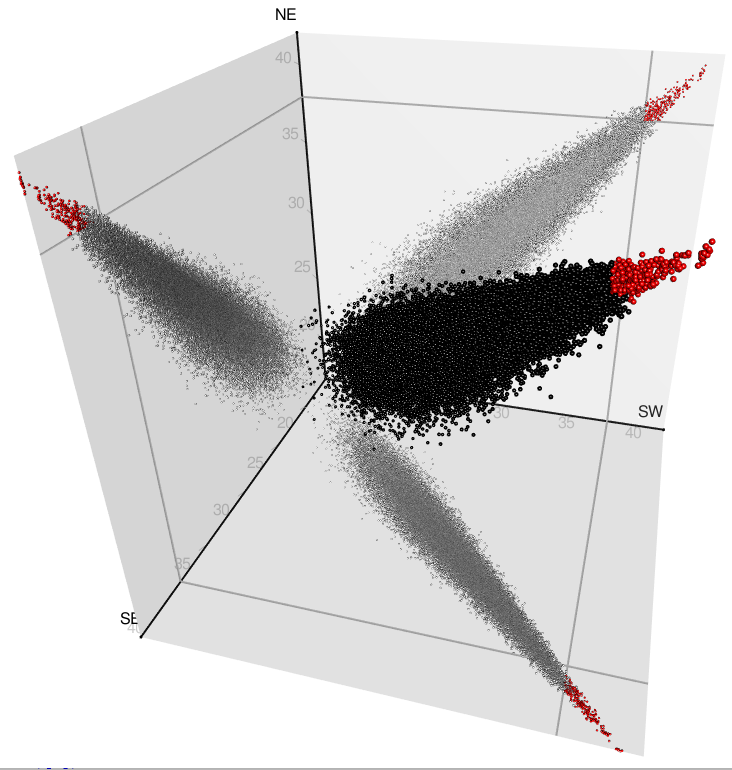}
  }
  \caption{\small Three inland regions (SW, SE, NE) and a scatterplot of simulated 60\,000 joint areal maxima across these regions. Points (and their bivariate projections) are marked red if all three areal maxima exceed $38^\circ$C.} 
  \label{fig:Q2}
\end{figure}

\paragraph{Caution!} All results in this section need to be interpreted with caution. First of all, one needs to account for the uncertainty inherent in the estimated characteristics, such as the ones asked for in (Q1) and (Q2). Provided that the fitted model was perfectly correct, these could be estimated with arbitrary precision by the use of a sufficiently large number of simulations. In practice, however, as summarized in \cite{dh2015}, there are different types of uncertainty in the model fit, including uncertainties related to taking measurements, choice of model class and parameter estimation, for instance.

The vigilant reader will have also noted, that already phrasing the initial questions (Q1) and (Q2) in this way comes along with several assumptions about the data and the underlying processes. Both (Q1) and (Q2) express that we assume to see a more or less homogeneous behaviour of maximum temperatures within a given summer and otherwise independent and identically distributed summers. All the more, we would like to stress that our (imperfect) answers defy easy interpretation in a climatological context that is far beyond the scope of this article.\vspace*{3mm}

Instead, one may understand this section as a simple analysis of the current climate in the sense of the studied questions and time frame. We deliberately draw attention to the spatial rather than temporal aspects here, which is in line with the core contents of this article. And we hope that there is no doubt left about the critical role of the simulations in this data example to provide answers to questions that otherwise would be difficult to address at all.

\section{A survey of simulation approaches}\label{sec:survey}

In this section, we will give an overview over existing algorithms for the 
simulation of a simple max-stable process $\{Z(\bm{x})\}_{\bm{x} \in K}$ on a compact
domain $K$. Almost all simulation approaches are based on the fact that
any sample-continuous simple max-stable process $Z$ possesses a 
\emph{spectral representation} 
\begin{align} \label{eq:spec-repr}
\{Z(\bm{x})\}_{\bm{x} \in K} = \bigg\{ \max_{j \in \NN} \Gamma_j^{-1} V_j(\bm{x})\bigg\}_{\bm{x} \in K} \quad \text{in distribution,}
\end{align}
where $\{\Gamma_j\}_{j \in \NN}$ are the arrival times of a unit rate Poisson
process on $(0,\infty)$ with independent markings $\{V_j\}_{j \in \NN}$ that 
are distributed according to a {non-negative} sample-continuous stochastic
process $\{V(\bm{x})\}_{\bm{x} \in K}$, the so-called \emph{spectral process} 
\citep{dh84,ghv90,penrose92}. Since $Z$ possesses standard unit Fr\'echet margins, the spectral process $V$ satisfies $\EE\{V(\bm{x})\} = 1$ for all $\bm{x} \in K$. Conversely, any continuous stochastic process $V$ that satisfies the moment condition $\EE\{V(\bm{x})\} = 1$ for all $\bm{x} \in K$ gives rise to a max-stable process $Z$ via \eqref{eq:spec-repr}. The representation \eqref{eq:spec-repr} is illustrated in Figure \ref{fig:maxstab-construction}.

\begin{figure}
	\centering
	\mbox{
		\includegraphics[height=5cm]{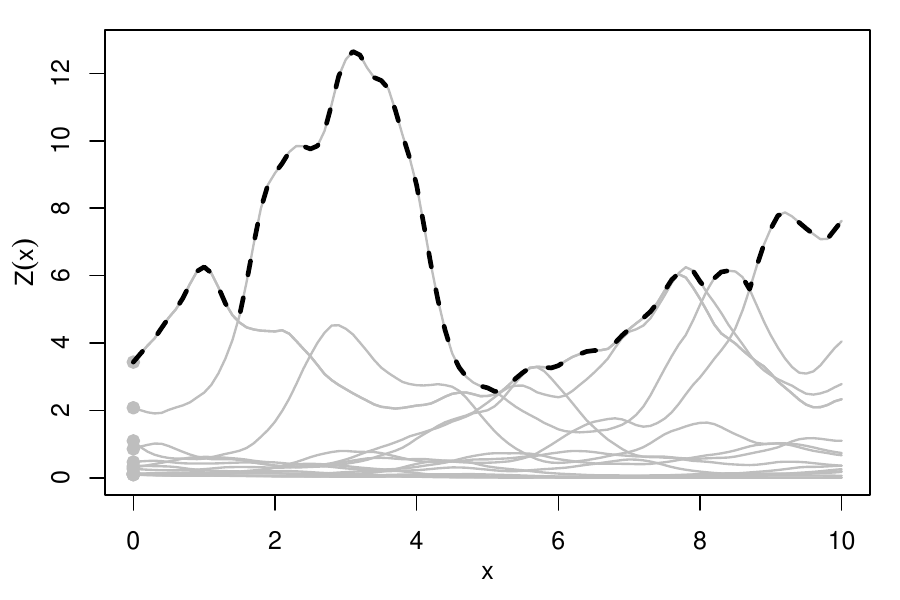}
	}
	\caption{\small Illustration of the spectral representation \eqref{eq:spec-repr}. Grey points represent the reciprocal arrival times $\Gamma_j^{-1}$, $j \in \NN$, while grey lines correspond to the processes $ \Gamma_j^{-1} V_j(\cdot)$, $j \in \NN$. The resulting max-stable
	process $Z$ is marked by the black dashed line.}
	\label{fig:maxstab-construction}
\end{figure}

It is important to note that the law of the max-stable process $Z$ in \eqref{eq:spec-repr} does not uniquely determine the law of the spectral process $V$. Instead, a different  spectral process $V'$ in \eqref{eq:spec-repr} may result in the same max-stable process $Z$. More precisely, two spectral processes $V$ and $V'$ generate the same max-stable process (in distribution) if and only if
\begin{align*}\EE\bigg\{\max_{k=1}^n a_k V({\bm x}_k)\bigg\} = \EE\bigg\{\max_{k=1}^n a_k V'({\bm x}_k)\bigg\}\end{align*}
for all $n \in \NN$, $a_k > 0$ ($k=1,\dots,n$), $\{{\bm x}_1,\dots,{\bm x}_n\} \subset K$ \citep{dH78}. We will call the processes $V$ and $V'$ \emph{equivalent spectral processes} in this case. For instance, multiplying $V$ with a positive random variable with unit expectation and independent of $V$ yields the same max-stable process $Z$. In practice, the choice of spectral process $V$ can have a major  effect on the accuracy and efficiency of a certain simulation algorithm. As a starting point, however, we assume that a max-stable process $Z$ is given by a specific choice of the spectral process $V$ despite the fact that there may be other  more convenient equivalent spectral processes $V'$ for $Z$.

\subsection{Threshold stopping -- the general idea} \label{sec:thresholdstopping}

The first simulation algorithm was introduced by \citet{schlather02} and is 
motivated by the fact that  the points $\{\Gamma_j\}_{j \in \NN}$ are the 
arrival times of a renewal process with standard exponential interarrival 
times. In particular, the points $\{\Gamma_j\}_{j \in \NN}$ are ordered: 
$\Gamma_1 < \Gamma_2 < \ldots$ almost surely. Therefore, we would expect that
the contribution of the process $\Gamma_j^{-1} V_j(\cdot)$ to the maximum in 
\eqref{eq:spec-repr} becomes smaller and smaller as $j$ gets large and at some
point negligible. In other words, the distribution of $Z$ can be approximated 
by the pointwise maximum  
\begin{align*}
Z^{(T)}(\bm{x}) = \max_{j=1,\ldots,T}  \Gamma_j^{-1} V_j(\bm{x}), \qquad \bm{x} \in K,
\end{align*}
where $T$ is a sufficiently large, but finite number. Instead of picking a     
deterministic number $T$, it turns out that an appropriately defined random 
stopping time $T$ results in more accurate approximations. Typically, a 
threshold dependent stopping time 
\begin{align} \label{eq:def-stop}
T=T_{\tau} = \min\Big\{j \in \NN: \ \Gamma_{j+1}^{-1} \tau < \inf_{\bm{x} \in K} Z^{(j)}(\bm{x})\Big\}  
\end{align}
is chosen, where $\tau > 0$ is a prescribed threshold reflecting an upper bound 
for the maximal contribution of the spectral process $V$ to the maximum in 
\eqref{eq:spec-repr}. A precise description of the sampling procedure is given 
by Algorithm~\ref{sec:survey}.\ref{alg:ts} below.

If the spectral process $V$ is uniformly bounded, we can choose $\tau$ large 
enough to satisfy $\sup_{\bm{x} \in K} V(\bm{x}) < \tau$ almost surely. 
Clearly, in this case, \eqref{eq:def-stop} implies that for all $\bm{x} \in K$
and $j > T_{\tau}$
\begin{align*} 
\Gamma_j^{-1} V_j(\bm{x}) < \Gamma_j^{-1} \tau \leq \Gamma_{T_{\tau}+1}^{-1} \tau  < Z^{(T_{\tau})}(\bm{x}).
\end{align*}
Consequently,
\begin{align*}
\{Z^{(T_{\tau})}(\bm{x})\}_{\bm{x} \in K} = \{Z^{(\infty)}(\bm{x})\}_{\bm{x} \in K} \quad \text{almost surely,}
\end{align*}
and a sample from the finite maximum $Z^{(T_{\tau})}$ can be seen as an exact
sample from $Z$, since the distribution of $Z^{(\infty)}$ equals the 
distribution of $Z$ by \eqref{eq:spec-repr}.

If, by contrast, $\PP\{\sup_{\bm{x} \in K} V(\bm{x}) > \tau\} > 0$, there is a 
positive probability that $Z^{(T_{\tau})}(\bm{x}) \neq Z^{(\infty)}(\bm{x})$ 
for some $\bm{x} \in K$. In this case, samples from $Z^{(T_{\tau})}$ only 
provide approximations to the process of interest $Z$.

\begin{algorithm}[caption={Threshold Stopping Algorithm.}, label={alg:ts}]
input: $\text{domain } K, \text{ threshold } \tau $
output: $\text{one (approximate/exact) max-stable process realization }
	   z \text{ on } K$
begin
	set $z(\bm{x}) = 0$ $\text{ for all}$ $\bm{x} \in K$
	simulate $\Gamma \sim \mathrm{Exp}(1)$
	while $\tau / \Gamma \geq \inf_{\bm{x} \in K} z(\bm{x})$
		simulate $v \sim V$
		set $z(\bm{x}) = \max\{ \Gamma^{-1} v(\bm{x}), z(\bm{x})\}$ $\text{ for all}$ $\bm{x} \in K$
		simulate $E \sim \mathrm{Exp}(1)$ 
		set $\Gamma = \Gamma + E$
	end while   
	return $z$
end
\end{algorithm}

\begin{remark}
	The stopping time in \eqref{eq:def-stop} is almost surely finite, since we 
	assumed the max-stable process $Z$ to be sample-continuous. Indeed, 
	sample-continuity implies that $\inf_{\bm x \in K} Z(\bm x) > 0$ almost surely 
	and only a finite number of functions 
	$\{\Gamma_j^{-1} V_j(\bm x)\}_{\bm x \in K}$, $j \in \NN$, contributes to the 
	maximum in \eqref{eq:spec-BR}, see \citet[Theorem 2.2]{de12} and 
	\citet[Corollary 9.4.4]{dhfeBook06} respectively. Therefore, the infimum of the
	$Z^{(j)}$'s on the right-hand side in \eqref{eq:def-stop} exceeds $0$ after a 
	finite number of steps $j$ almost surely, while the inverses of the $\Gamma_j$'s
	on the left-hand side tend to $0$ with probability one. Consequently, the
	stopping time in \eqref{eq:def-stop} is almost surely finite. 
\end{remark}

\subsection{Threshold stopping -- normalizing spectral processes} \label{sec:normalized}

As discussed above, Threshold Stopping Algorithm~\ref{sec:survey}.\ref{alg:ts}
provides exact realizations of the max-stable process $Z$ if the spectral 
process $V$ is almost surely bounded. If this is not the case for $V$, it can
still often be transformed into an equivalent spectral process $V'$, which is uniformly 
bounded and generates the same max-stable process $Z$ in the sense of  \eqref{eq:spec-repr}.
 Two such procedures have been studied in further detail, both of which
transform a given spectral process $V$ into an equivalent spectral process 
$V^{\lVert\cdot\rVert}$, which is normalized w.r.t.\ some norm 
$\lVert\cdot\rVert$, i.e.\ it satisfies
\begin{align}\label{eq:normalized}
\lVert V^{\lVert\cdot\rVert}\rVert = \theta_{\|\cdot\|} \qquad
\text{almost surely for some }\theta_{\|\cdot\|}>0.
\end{align}
The constant $\theta_{\|\cdot\|}$ is uniquely determined by 
\begin{align} 
\theta_{\|\cdot\|}= \EE  \|V\| = \lim_{u \to \infty} u \, \PP\{ \|Z\| > u\}.
\end{align}
and does not depend on the choice of the starting spectral process $V$.

\begin{remark} \label{rk:pareto}
  More generally, \cite{dombry-ribatet-15} show that for each sample-continuous simple max-stable process $Z$ and each non-negative measurable 1-homo\-gen\-eous functional $\ell$ on $C(K,[0,\infty))$, there exists a spectral process $V^\ell$ for $Z$ in the sense of \eqref{eq:spec-repr} that is uniquely characterized by the property $\ell(V^\ell)=\theta_\ell$ a.s.\ for some constant $\theta_\ell > 0$. The constant $\theta_\ell$ is necessarily given by $\theta_\ell=\EE \ell(V)=\lim_{u\to \infty}u\PP\{\ell(Z)>u\}$ and we may call $V^\ell$ the $\ell$-normalized spectral process of $Z$. If $V$ is an arbitrary spectral process for the simple max-stable process $Z$, the equivalent $\ell$-normalized spectral process $V^\ell$ can be obtained by a measure transform $V^\ell = \theta_\ell Y/\ell(Y)$, where $\PP\{Y \in dv\}=\ell(v)/\theta_\ell\,\PP\{V \in dv\}$. The $\ell$-normalized spectral process $V^\ell$  characterizes extremes of stochastic processes also in terms of threshold exceedances instead of maxima. Let $X$ be a sample-continuous process in the max-domain of attraction of $Z$. Then, as $u\uparrow \infty$, the conditional distribution of $u^{-1} X$ given that $\ell(X) > u$ converges weakly to the distribution of the product $P \cdot \theta_{\ell}^{-1} V^{\ell}$, where $P$ is a standard Pareto random variable independent of the process $V^{\ell}$. Thus, the resulting limit process, the so-called $\ell$-Pareto process \citep{dombry-ribatet-15}, is fully described by the $\ell$-normalized spectral process $V^{\ell}$ and being able to effectively simulate from $V^{\ell}$ has important implications beyond the max-stable context. In what follows, with the max-stable simulation context in view, we consider the cases $\ell=\lVert \cdot \rVert_\infty$ (when $K$ is compact) and $\ell=\lVert \cdot \rVert_1$ and $\ell$ being the evaluation at a single point in $K$ (when $K$ is a finite set) and discuss how the resulting processes $V^\ell$ are related to each other.
\end{remark}

\paragraph{Sup-normalization.} The first transformation of this type was 
introduced in \citet{osz18} who proposed a normalization w.r.t.\ the supremum
norm $\lVert f \rVert_\infty = \sup_{\bm{x} \in K} f(\bm{x})$ for $f\geq 0$. 
Starting from an arbitrary sample-continuous spectral process $V$, the unique
equivalent sup-normalized spectral process $V^{\lVert\cdot\rVert_\infty}$, which satisfies 
$\sup_{\bm{x} \in K} V^{\lVert\cdot\rVert_\infty}(\bm{x}) {=} \theta_{\|\cdot\|_\infty}$
almost surely, can be obtained as the normalization
\begin{align} \label{eq:SN-Y}
V^{\lVert\cdot\rVert_\infty}(\bm{x}) =  \frac{\theta_{\|\cdot\|_\infty} Y(\bm{x})}{\sup_{\bm{x}' \in K} Y(\bm{x}')} \end{align}
of a stochastic process $Y$ with transformed law 
\begin{align*} 
\PP\{Y \in A\} = \frac{1}{\theta_{\|\cdot\|_\infty}} \int_A \sup_{\bm{x} \in K} v(\bm{x}) \, \PP\{V \in \dd v\}, \quad A \subset C(K).
\end{align*}

As the spectral processes $V$ and $V^{\lVert\cdot\rVert_\infty}$ are equivalent, we can use the
sup-normalized spectral process for simulation.  By construction, all the sample 
paths of $V^{\lVert\cdot\rVert_\infty}$ are bounded by 
$\theta_{\|\cdot\|_\infty}$ almost surely and, by \citet{rr91}, sample-continuity of $Z$ already implies 
that $\theta_{\|\cdot\|_\infty}$ is finite 
\citep[see also][Theorem 9.6.1]{dhfeBook06}. Therefore, the output of 
Algorithm~\ref{sec:survey}.\ref{alg:ts} with threshold 
$\tau=\theta_{\|\cdot\|_\infty}$ is an exact realization of the max-stable 
process $Z$, when the sup-normalized spectral representation 
$V^{\lVert\cdot\rVert_\infty}$ is used therein. 

By \eqref{eq:SN-Y} simulation of $V^{\lVert\cdot\rVert_\infty}$ can be based on simulation of the transformed process~$Y$. While \citet{osz18} suggest an
approximating Markov Chain Monte Carlo (MCMC) algorithm for this task, more
recently, \citet{defondeville-davison-18} present a relation that allows for exact simulation of $Y$ via rejection sampling provided that a simulation procedure
for the normalized spectral process $V^{\|\cdot\|}$ for an arbitrary norm is 
given. We refer to Section~\ref{sec:BRgeneric} for further efficiency improvements when $V$ is log-Gaussian. The resulting process $Y/\lVert Y \rVert_\infty$ has then the law of $V^{\lVert\cdot\rVert_\infty}/\theta_{\|\cdot\|_\infty}$. While analytic expressions for the normalizing constant $\theta_{\|\cdot\|_\infty}$ are usually not available, it can still be estimated in the course of the simulation procedure, e.g.\ via the relation $\theta_{\|\cdot\|_\infty} = \EE\|V\|_\infty$, and can subsequently be used for an ex-post normalization. 
The constant $\theta_{\|\cdot\|_\infty}$ is also known as \emph{extremal coefficient}, cf.~\eqref{eq:extremal-coefficient}. \vspace*{3mm}

\paragraph{Sum-normalization.} The second transformation of type 
\eqref{eq:normalized} uses a normalization w.r.t.\ the $\ell_1$-norm 
$\lVert f \rVert_1 = \sum_{k=1}^N f(\bm{x}_k)$ for $f \geq 0$ on a finite 
domain $K = \{\bm{x}_1,\ldots,\bm{x}_N\}$. It has been proposed by \cite{dm15}
for the special case of Brown-Resnick processes (see Section~\ref{sec:BR}) and
extended to a more general framework by \cite{deo16}. The starting point 
for the construction of this representation is the fact that, given the distribution of $Z$, for each 
$k \in \{1,\ldots,N\}$, there is a unique equivalent spectral process $V^{(k)}$ of $Z$ with the property
$V^{(k)}(\bm{x}_k)=1$ almost surely. Its law is given by 
\begin{align} \label{eq:measure-Pk} 
P_k(A) = \int_{[0,\infty)^K} v(\bm{x}_k) \, \mathbf{1}{\{ v \in v(\bm{x}_k) A\}}  \,\, \PP\{V \in \dd v\}, 
\end{align}
for $A \subset [0,\infty)^K = [0,\infty)^{\{\bm{x}_1,\ldots,\bm{x}_N\}}$,
where $V$ is an arbitrary spectral process of $Z$. \citet{deo16} show that the 
unique equivalent sum-normalized spectral process $V^{\lVert \cdot \rVert_1}$, which 
satisfies $\sum_{k=1}^N V^{\lVert \cdot \rVert_1}(\bm{x}_k) = 
\theta_{\lVert\cdot\rVert_1} = N$ almost surely, is then given by
\begin{align*} 
V^{\lVert\cdot\rVert_1} = N \frac{Y}{\lVert Y\rVert_1}, \quad \text{where } Y \sim \frac 1 N \sum_{k=1}^N P_k.
\end{align*}

Since $V^{\lVert\cdot\rVert_1}$ satisfies $\lVert V^{\lVert\cdot\rVert_1}\rVert_1 = N$ almost surely, and each component of
a vector is bounded by its $\ell_1$-norm, that is, 
$V^{\lVert \cdot \rVert_1}(\bm x) \leq \lVert V^{\lVert \cdot \rVert_1}\rVert_1 = N$ for all $\bm x \in K$,
the sum-normalized spectral process $V^{\lVert\cdot\rVert_1}$ can be used as 
spectral process for Algorithm~\ref{sec:survey}.\ref{alg:ts} resulting in exact 
realizations for the threshold $\tau=N$. \citet{deo16} also explicitly 
calculate the laws $P_k$, $k=1,\ldots,N$, and, thus, verify that they can be
easily sampled for many popular max-stable models such as Brown-Resnick 
processes or extremal-$t$ processes.

\subsection{Extremal functions} \label{sec:EF}

A simulation procedure that essentially differs from the previously considered 
threshold stopping algorithm is the extremal functions approach, which was also
introduced in \cite{deo16}. Instead of simulating  the elements of the Poisson
point process $\Phi = \{\varphi_j\}_{j \in \NN}$ with $\varphi_j(\cdot) =
\Gamma_j^{-1} V_j(\cdot)$ in an ascending order w.r.t.\ $\Gamma_j$ until a 
stopping criterion is fulfilled, the idea is to simulate only the so-called
\emph{extremal functions} \citep{de13} that definitely contribute to the final
maximum in \eqref{eq:spec-repr} on the finite domain $K=\{\bm{x}_1,\ldots,\bm{x}_N\}$, i.e.\ all the functions $\varphi \in \Phi$ 
such that
\begin{align} \label{eq:extremal-fctn}
\varphi(\bm{x}_k) = \max_{j \in \NN} \varphi_j(\bm{x}_k)
\end{align}
for some $k \in \{1,\ldots,N\}$. It can be shown that, for each 
$k \in \{1,\ldots,N\}$, with probability one, there is exactly one 
\emph{(extremal) function} $\varphi \in \Phi$ satisfying 
\eqref{eq:extremal-fctn}, which we denote by $\varphi^{(k)}_+$ in the 
following. The algorithm subsequently simulates the processes
\begin{align} \label{eq:ef-intermediate}
Z^{(k)}_+(\cdot) = \max_{j=1,\ldots,k} \varphi^{(j)}_+(\cdot).
\end{align}
By construction, the process $Z^{(k)}_+$ is exact at $\bm{x}_1,\ldots,\bm{x}_{k}$, i.e.\
\begin{align*} 
\{Z^{(k)}_+(\bm{x}_i)\}_{i=1,\ldots,k} = \{Z(\bm{x}_i)\}_{i=1,\ldots,k} \quad \text{in distribution}. 
\end{align*}
In particular, the final process $Z^{(N)}_+$ has the same distribution as the 
desired max-stable process $Z$ on the entire domain $K=\{\bm{x}_1,\ldots,\bm{x}_N\}$.

The theory behind this procedure stems from \cite{de13} who show --  by using Slivnyak-Mecke type arguments -- that, for each $k=1,\ldots,N$, the point
process $\Phi \setminus \{\varphi^{(1)}_+, \ldots, \varphi^{(k)}_+\}$ is conditionally independent from $ \{\varphi^{(1)}_+, \ldots, \varphi^{(k)}_+\}$
conditional on $Z(x_1),\ldots,Z(x_k)$. Based on this result,
the law of the collection of all $N$ needed extremal functions $\{\varphi^{(k)}_+\}_{k=1,\dots,N}$ can be described by an iterative procedure:
\begin{itemize}
\item Initially, $\varphi^{(1)}_+ \sim \Gamma_1^{-1} V^{(1)}$.
  \item
    Subsequently, the conditional law of the next extremal function $\varphi^{(k)}_+$, when the previous extremal functions $\varphi^{(1)}_+,\ldots,\varphi^{(k-1)}_+$ are already given,   can be described as follows:
    \begin{enumerate}[label={(\roman*)},itemsep=0mm]
    \item Either $\varphi^{(k)}_+$ is the unique argmax of the evaluation functional at $\bm{x}_k$ in a Poisson point process $\Phi^{(k)}$ on $[0,\infty)^K$, whose intensity measure equals the intensity measure of $\Phi$ restricted to the set
      \begin{align*}
        \{\varphi \in [0,\infty)^K: \ \varphi(\bm{x}_i) < Z^{(k-1)}_+(\bm{x}_i) \text{ for } i<k \ \text{ and } \
          \varphi(\bm{x}_k) >  Z^{(k-1)}_+(\bm{x}_k)
          \}
          \end{align*}
      \item Or, in the event that $\Phi^{(k)}$ is empty, $\varphi^{(k)}_+$ is one of the previous extremal functions.
    \end{enumerate}  
\end{itemize}
The restricted Poisson point process $\Phi^{(k)}$ in (i) can be conveniently simulated by iteratively generating Poisson points $\varphi^{(k)}_j = \Gamma_j^{-1} V_j$ from the original Poisson point process $\Phi$, where we choose the spectral processes $V_j$ as independent copies of the special spectral process $V^{(k)}$ (cf.~\eqref{eq:measure-Pk}). This ensures that $\varphi^{(k)}_j(\bm x_k)=\Gamma_j^{-1} V_j(\bm x_k) =  \Gamma_j^{-1}$ a.s. Thereby, similarly to a threshold stopping procedure, the potential new value at ${\bm x_k}$ is running through the descending values $\Gamma_1^{-1} > \Gamma_2^{-1} > \dots$. However, since we are ultimately interested in $\Phi^{(k)}$ (and not $\Phi$), we test each time if $\varphi^{(k)}_j(\bm{x}_i) < Z^{(k-1)}_+(\bm{x}_i)$ for all $i<k$. If that happens, we have found our new argmax within $\Phi^{(k)}$ and hence extremal function $\varphi^{(k)}_+=\varphi^{(k)}_j$. Otherwise, we will arrive at a $\Gamma_j^{-1}$ that falls below $Z^{(k-1)}_+(\bm{x}_k)$ for some $j\geq 1$. This corresponds to the event that $\Phi^{(k)}$ is empty.

The entire procedure is summarized by 
Algorithm~\ref{sec:survey}.\ref{alg:ef} below. \vspace*{3mm}

\begin{algorithm}[caption={Extremal Functions Algorithm.}, label={alg:ef}]
input: $\text{domain}$ $K=\{\bm{x}_1,\ldots,\bm{x}_N\}$ 
output: $\text{one (exact) max-stable process realization } z \text{ on } K$
begin
	set $z(\bm{x}_k) = 0$ $\text{ for}$ $k=1,\ldots,N$
	for $k = 1, \ldots, N$
		simulate $\Gamma \sim \mathrm{Exp}(1)$
		while $\Gamma^{-1} > z(\bm{x}_k)$
			simulate $v \sim P_k$
			if $k=1$ or $\Gamma^{-1} v (\bm{x}_i) < z(\bm{x}_i)$ $\text{ for all}$ $i < k$  
				set $z(\bm{x}_i) = \max\{ \Gamma^{-1} v(\bm{x}_i), z(\bm{x}_i)\}$ $\text{ for all}$ $i \geq k$
			end if 
			simulate $E \sim \mathrm{Exp}(1)$ 
			set $\Gamma = \Gamma + E$ 
		end while 
	end for   
	return $z$
end
\end{algorithm}

\subsection{Summary of generic simulation approaches} \label{sec:summary-generic}

The aforementioned simulation approaches to obtain a simple max-stable process $Z$ are generic in the sense that they are not tailored to a specific class of max-stable processes. Figure~\ref{fig:generic-summary} provides a quick overview. We would like to draw attention to the fact that each of these approaches relies on the ability to simulate from a (family of) specific spectral process(es) $V$ for the max-stable process $Z$, cf.\ also Table~\ref{table:efficiency}. Considering a finite domain $K=\{\bm x_1,\ldots,\bm x_N\}$, the extremal functions approach needs the simulation of spectral processes $V^{(1)},\dots,V^{(N)}$ (i.e.\ simulation from the $N$ measures in \eqref{eq:measure-Pk}) readily available. The sum-normalized and sup-normalized approaches are by definition based on the availabilty of the spectral processes $V^{\lVert\cdot\rVert_1}$ and $V^{\lVert\cdot\rVert_\infty}$, respectively. As explained in Section~\ref{sec:normalized}, availability of $V^{(1)},\dots,V^{(N)}$ guarantees availability of $V^{\lVert\cdot\rVert_1}$ at a negligible additional computational cost (drawing a point from the finite domain $K=\{\bm{x}_1,\ldots,\bm{x}_N\}$). Further, if any normalized $V^{\lVert\cdot\rVert}$ is available (for instance, the sum-normalized spectral process $V^{\lVert\cdot\rVert_1}$), then it can also be used to simulate the sup-normalized spectral functions $V^{\lVert\cdot\rVert_\infty}$ via rejection sampling up to a multipicative constant \citep{defondeville-davison-18}, cf.\ Section~\ref{sec:normalized}. Rejection sampling is usually costly; the computational cost can be reduced in some special cases, cf.\ Section~\ref{sec:BRgeneric}.

\begin{figure}
  \small
  \centering
  \begin{tcolorbox}
    \hspace*{-4mm}
  \begin{tikzpicture}[inner sep=1mm]
    \tikzstyle{level 1}=[sibling distance=70mm, level distance = 2.2cm]
\tikzstyle{level 2}=[sibling distance=40mm, level distance = 3.2cm]
\node {\begin{minipage}{7cm}\centering
    \textbf{Generic algorithms}\\  
    \end{minipage}
}
child {node {
    \begin{minipage}{6cm}
      \centering
    \textbf{Threshold stopping}\\
           {\footnotesize (\cite{schlather02}; Section~\ref{sec:thresholdstopping})}
           { \begin{align*}Z  = \max_{j=1}^{T_\tau}  \, {\Gamma_j^{-1} V_j}\end{align*}}
           \vspace*{-\baselineskip}
 \end{minipage}}
    child {node {
      \begin{minipage}{4cm}
        \centering
                 \textbf{Sup-normalization\\ (SN)}\\
             { \footnotesize (\cite{osz18}; Section~\ref{sec:normalized})}
             { \begin{align*}V_j \sim V^{\lVert \cdot \rVert_\infty}\end{align*}}
                 \vspace*{-\baselineskip}
    \end{minipage}}}
  child {node {
      \begin{minipage}{4cm}
        \centering
         \textbf{Sum-normalization\\ (DM)}\\
             { \footnotesize (\cite{dm15}; Section~\ref{sec:normalized})}
             {\begin{align*}V_j \sim V^{\lVert \cdot \rVert_1}\end{align*}}
                 \vspace*{-\baselineskip}
      \end{minipage}
  }}
}
child {node[inner sep=-7mm] {\hspace*{-10mm}
    \begin{minipage}{6cm}
      \centering
      \vspace*{\baselineskip}
      \vspace*{\baselineskip}
      \textbf{Extremal functions (EF)}\\
      {\footnotesize (\cite{deo16}; Section~\ref{sec:EF})}
      { \begin{align*}Z = \max_{k = 1}^N  \, {\varphi^\text{(k)}_+}\end{align*} }
    \end{minipage}
}};
  \end{tikzpicture}
  \end{tcolorbox}
  \caption{\small Overview over generic simulation approaches to obtain a max-stable process $Z$ from one of its spectral processes $V$.}
  \label{fig:generic-summary}
  \end{figure}

\section{Specialties for Brown-Resnick processes} \label{sec:BR}

Among several popular max-stable processes, the class of Brown-Resnick 
processes stands out as a parsimonious spatial model that has become a 
benchmark in spatial extremes. Let $\{W(\bm{x})\}_{\bm{x} \in K}$ be a centered
Gaussian process with variance $\{\sigma^2(\bm{x})\}_{\bm{x} \in K}$. Then the 
max-stable process $\{Z(\bm{x})\}_{\bm{x} \in K}$ that is associated to the 
spectral process 
\begin{align}\label{eq:spec-BR}
V(\bm{x})=\exp\bigg(W(\bm{x}) - \frac{\sigma^2(\bm{x})}{2}\bigg), \quad \bm{x} \in K
\end{align}
via \eqref{eq:spec-repr} has unit Fr\'echet marginal distributions and its law
depends only on the \emph{variogram} 
\begin{align*}
\vario: K \times K \rightarrow [0,\infty), \qquad \vario(\bm{x},\bm{y})=\EE(W(\bm{x})-W(\bm{y}))^2
\end{align*}
\citep{kab11}. In particular, for $K \subset \RR^d$, the max-stable process $Z$
is stationary if the variogram $\vario$ depends only on $\bm{x}-\bm{y}$ and by 
slight abuse of notation we write $\vario(\bm{x}-\bm{y})=\vario(\bm{x},\bm{y})$
in this case. The stationary process $Z$ has first been introduced in 
\cite{ksdh09} in this generality and is now widely known as 
\emph{Brown-Resnick process}. In practice, {among unbounded variograms on $\RR^d$},
it is almost exclusively the variogram family 
$\vario(\bm{x}-\bm{y})=\lVert (\bm{x} - \bm{y})/s \rVert^\alpha$, $s>0$, 
$\alpha \in (0,2)$ of fractional Brownian sheets (fBS) that is considered in 
applications. 

\subsection{Threshold stopping based on Gaussian mixtures}
The first attempts of simulating a Brown-Resnick process $Z$ were based on 
threshold stopping with a log-Gaussian spectral process $V$ as in 
\eqref{eq:spec-BR} satisfying $W(\bm{x}_o)=0$ for some $\bm{x}_o\in K$. Typically, the
origin $\bm{o} \in \RR^d$ belongs to the simulation domain $K$ and 
$\bm{x}_o=\bm{o}$. We refer to such spectral processes $V=V^{\text{(orig)}}$ as
the \emph{original} spectral representation of the Brown-Resnick process $Z$. 
Since log-Gaussian processes do not have an almost sure upper bound, such a 
threshold stopping procedure based on $V^{\text{(orig)}}$ cannot be exact. 
Instead, a typical phenomenon is that the threshold stopping procedure works 
well in a neighbourhood of $\bm{x}_o$, where the variance of the underlying 
Gaussian process is small, but a simulation bias appears in those parts of the
domain where the variance is large.
To avoid this phenomenon, \cite{oks12} introduced a uniformly distributed 
\emph{random shift} in the spectral process
\begin{align}\label{eq:randomshift}
V^{\text{(shift)}}(\bm{x})=V^{\text{(orig)}}(\bm{x}-\bm{S}), \quad \bm{x} \in K, \quad \bm{S} \sim \text{Unif}(K).
\end{align}
Note that the superimposed homogeneity comes however at the cost of increasing
the variance of the spectral process even further in most situations. 

More recently, \cite{os18} explain how a variance reduction of the Gaussian 
process $W$ in \eqref{eq:spec-BR} with fixed variogram $\vario$ can lead to 
faster and more accurate simulations based on the threshold stopping procedure.
Specifically, when $W$ is chosen such that the maximal variance 
$\sup_{\bm{x} \in K} \Var(W(\bm{x}))$ is minimal among all Gaussian processes 
on $K$ with variogram $\vario$, we call the corresponding spectral process in
\eqref{eq:spec-BR} \emph{minimal variance} spectral process 
$V^{\text{(minvar)}}$. Table~\ref{table:minvar} lists the corresponding minimal
Gaussian processes on the $d$-dimensional hyperrectangle 
${[-\bm{R},\bm{R}] \subset \RR^d}$ for the variogram family 
$\vario(\bm{h})=\lVert \bm{h}/s \rVert^\alpha$, $\alpha \in (0,2)$, $s>0$. For 
$d \geq 2$ and $\alpha \in (0,1)$ the minimal representation is unknown. 
However, also in this case the modified Gaussian process 
\begin{align*}
W(\bm{x})=W^{\text{(orig)}}(\bm{x}) - 2^{-d} \sum_{\bm{v} \in \Ex([-\bm{R},\bm{R}])} W^{\text{(orig)}}(\bm{v}),
\end{align*}
where $\Ex([-\bm{R},\bm{R}])$ is the set of vertices of the simulation domain 
$[-\bm{R},\bm{R}] \subset \RR^d$, has a substantially reduced maximal variance 
compared to the original process $W^{\text{(orig)}}$ and should be preferred.

\begin{table}
	\caption{\small  
		Gaussian process $W^{\text{\normalfont (minvar)}}$ with variogram $\vario(\bm{h})=\lVert \bm{h}/s \rVert^\alpha$, $\alpha \in (0,2)$, 
		$s>0$ that minimize the maximal variance on the hyperrectangle $[-\bm{R},\bm{R}] \subset \RR^d$.
		The process is either given by its covariance $C^{\text{\normalfont (minvar)}}$
		or built from the original representation $W^{\text{\normalfont (orig)}}$. 
	}
	\label{table:minvar}
	\vspace*{1mm}
	\centering
	\tabcolsep2mm
	\begin{tabular}{l|c|c}
		\hline
                \rowcolor{gray!30}
		& $d=1$ & $d \geq 2$\\
		\hline 
		$\alpha \in (0,1]$
		& 
		\begin{minipage}{9.4cm}
			\centering
			\vspace*{2mm}
			${C^{\text{(minvar)}}(x,y)={2^{-1}s^{-\alpha}}\big({\Gamma\big(\frac{2-\alpha}{2}\big)\Gamma\big(\frac{1+\alpha}{2}\big)}{\Gamma\big(\frac{1}{2}\big)^{-1}}-\lvert x-y\rvert^\alpha \big)}$ 
			\vspace*{2mm}
		\end{minipage}
		& 
		\begin{minipage}{2cm}
			\centering
			\vspace*{3mm}
			unknown
			\vspace*{3mm}
		\end{minipage}
		\\
		\hline & \multicolumn{2}{c}{} \\[-3mm]
		$\alpha \in [1,2)$ &
		\multicolumn{2}{c}{
			\begin{minipage}{11cm}
				\centering
				\vspace*{2mm}
				$
				W^{\text{(minvar)}}(\bm{x})
				= W^{\text{(orig)}}(\bm{x}) - 2^{-d} \sum_{\bm{v} \in \Ex([-\bm{R},\bm{R}])} W^{\text{(orig)}}(\bm{v})
				$
				\vspace*{2mm}
			\end{minipage}
		}\\
		\hline
	\end{tabular}
\end{table}

\subsection{Record breakers} \label{subsec:record}

An exact simulation procedure for Brown-Resnick processes, which is 
specifically tailored to this class, is the record breakers approach by 
\citet{lbdm16}. It is based on the original spectral representation 
\eqref{eq:spec-repr} with $V$ being a log-Gaussian random field as in 
\eqref{eq:spec-BR}. Let $a, c \in (0,1)$ and $C>0$ and consider the three 
random times 
\begin{align*}
N_X ={}& \sup\{n \in \NN: \, \max_{i=1,\ldots,N} V_n(\bm x_i) > n^a \exp(C) \}, \\
N_\Gamma ={}& \sup\{n \in \NN: \, \Gamma_n \leq cn \}, \\
N_a ={}& \sup\Big\{n \in \NN: \, n c \leq \frac{\Gamma_1 n^a \exp(C)}{\min_{i=1,\ldots,N} V_1(\bm x_i)} \Big\} 
= \Big\lfloor\Big(\frac{\Gamma_1 \exp(C)}{c \min_{i=1,\ldots,N} V_1(\bm x_i)}\Big)^{1/(1-a)}\Big\rfloor.
\end{align*}
From the definition of $N_X$, $N_\Gamma$ and $N_a$, it is easily checked that
\begin{align*}
\{Z^{(\max\{N_X,N_\Gamma,N_a\})}(\bm{x})\}_{\bm{x} \in K} = \{Z^{(\infty)}(\bm{x})\}_{\bm{x} \in K} \quad \text{ in distribution.}
\end{align*}
While $N_a$ can be obtained directly from $\Gamma_1$ and $V_1$, simulation of 
the random times $N_X$ and $N_\Gamma$ is more sophisticated. For the simulation
of $N_\Gamma$, \citet{lbdm16} make use of the fact that $\{\Gamma_n - cn\}_{n \in \NN}$
is a random walk with positive drift. An algorithm is provided that 
subsequently samples the times when the random walk crosses zero. Due to its 
positive drift the process will finally stay positive. To obtain $N_X$, all 
so-called \emph{record-breaking} times $\eta_1 < \eta_2 < \ldots$, i.e.\ all
$\eta \in \NN$ such that
\begin{align*} \max_{i=1,\ldots,N} V_\eta(\bm x_i) > n^a \exp(C), \end{align*}
are subsequently simulated. 
The finiteness of all moments of $\max_{i=1,\ldots,N} V(\bm x_i)$ implies
that the number of record-breaking times is almost surely finite. 
Consequently, the record-breakers algorithm requires an almost surely finite
number of simulations of log-Gaussian processes $V_i$ to obtain a realization 
of the Brown-Resnick process $Z$.

\begin{remark}
  On a practical note, \cite{lbdm16} also provide guidance on how to choose the parameters $a,c \in (0,1)$ and $C > 0$. However, there is an additional parameter $\delta \in (0,1)$ involved, where the practical implications of this choice and its interplay with the other parameters are still open.
\end{remark}

\subsection{Generic approaches}
\label{sec:BRgeneric}
Besides these approaches that are rather specific to the class of 
Brown-Resnick processes, general procedures such as simulation based on 
normalized spectral processes (Section~\ref{sec:normalized}) or the extremal
functions  approach (Section~\ref{sec:EF}) can be used, cf.\ also Figure~\ref{fig:generic-summary}.
Considering a finite domain $K=\{\bm x_1,\ldots,\bm x_N\}$, the distribution $P_k$ from \eqref{eq:measure-Pk} is the distribution of the stochastic process $\{\exp(W^{\text{(orig)}}(\bm x - \bm x_k))\}_{{\bm x} \in K}$ \citep{deo16}.
Therefore, the extremal functions approach as in Algorithm~\ref{sec:survey}.\ref{alg:ef}  is readily available for Brown-Resnick processes. Further, this implies that the sum-normalized spectral process is of the form
\begin{align} \label{eq:dm}
V^{\lVert\cdot\rVert_1}(\bm x) = N \frac{\exp(W^{\text{(orig)}}(\bm x - \bm S))}{\sum_{k=1}^N \exp(W^{\text{(orig)}}(\bm x - \bm x_k))}, \quad \bm x \in K,
\end{align}
where ${\bm S}$ is uniformly distributed on $K = \{\bm x_1,\ldots,\bm x_N\}$ 
and independent of $W^{\text{(orig)}}$, which has been demonstrated already earlier in \citet{dm15}. Finally, the simple representation \eqref{eq:dm} of the sum-normalized spectral functions can also be used to simulate the sup-normalized spectral functions via rejection sampling \citep{defondeville-davison-18}. Modifications of the last approach to reduce the rejection rate and alternative MCMC procedures have recently been proposed  by \citet{oss19}.

\begin{remark}
Furthermore, \citet{ho-dombry-17} show that, conditional on the component
$k^* \in \{1,\ldots,N\}$ where the maximum is assumed, the distribution of 
the vector 
$(V^{\lVert\cdot\rVert_\infty}(\bm x_k))_{k=1,\ldots,N} / V^{\lVert\cdot\rVert_\infty}(\bm x_{k^*})$ equals the 
distribution of a log-Gaussian vector conditional on not exceeding one -- a fact that can be used for its simulation. 
As efficient sampling from such a conditional distribution is not straightforward in high dimension and the calculation of the distribution of $k^*$ involves the inversion of several matrices of sizes $N \times N$ and $(N-1) \times (N-1)$ as well as evaluations of $(N-1)$-dimensional Gaussian distribution functions, however, this approach is limited to small or moderate $N$ in practice.
\end{remark}

\section{Desirable properties} \label{sec:eff-acc}

Simulation algorithms are supposed to provide results in an efficient and 
accurate way. In this section, we investigate the performance of the algorithms
above with respect to these two aspects from a theoretical angle. While the efficiency of an algorithm  can be characterized in terms of its computational complexity, we measure its accuracy in terms of distributional properties of the approximation error. The proofs for this section are postponed to the Appendix~\ref{app:proofs}.

\subsection{Efficiency} \label{sec:eff}

Apart from the record breakers approach (Section~\ref{subsec:record}), which is
tailored to the class of Brown-Resnick processes, all other simulation 
algorithms reviewed in this manuscript are based on the simulation of standard 
Poisson points $\Gamma_j$ on the positive real line and associated spectral 
processes $V_j$ on the simulation domain $K$ only. Hence, if $c_V(K)$ denotes 
the computational complexity of simulating a single spectral process $V$ on the
domain $K$ and $N_V(K)$ denotes the total number of spectral processes $V_j$ to 
be simulated in such a simulation algorithm, then the law of the product 
$N_V(K) \cdot c_V(K)$ describes the computational complexity of the entire 
procedure. As the second factor $c_V(K)$ inevitably depends on the simulation 
technique used to generate samples from the specific spectral function $V$, we 
focus our analysis mainly on the first factor $N_V(K)$ henceforth.

In case of the Threshold Stopping Algorithm~\ref{sec:survey}.\ref{alg:ts}, the 
random number $N_V(K)$ of spectral processes to be simulated coincides with the 
stopping time $T=T_{\tau}$ from \eqref{eq:def-stop}. Its expected value can be 
bounded as follows.   

\begin{proposition} 
	\label{prop:expected_stopping}
	\begin{enumerate}[label={(\alph*)},wide=0mm]
		\item The expected stopping time of the Threshold Stopping 
		Algorithm~\ref{sec:survey}.\ref{alg:ts} is bounded from below by
		\begin{align} \label{eq:expectedtime}
		\EE(N_V(K)) =\EE(T_{\tau}) \geq \tau \, \EE \, \Big\{ {1}/{\inf_{\bm{x} \in K} Z(\bm{x})} \Big\}.
		\end{align}
		\item Equality in \eqref{eq:expectedtime} holds if and only if 
		$\sup_{\bm x \in K} V(\bm x) \leq \tau$ almost surely.
	\end{enumerate}
\end{proposition}

The lower bound in \eqref{eq:expectedtime} is finite for sample-continuous $Z$
by Theorem 2.2 in \citet{de12}. 
It should be relatively sharp
in most practically relevant situations,  
while an ad-hoc
rough upper bound is given by
\begin{align*} 
\EE(N_V(K)) = \EE(T_{\tau}) \leq 1 + \EE\Big\{ {\tau}/{\inf_{\bm{x} \in K} V(\bm{x})} -1 \Big\}_+. 
\end{align*}
Another interpretation of Proposition~\ref{prop:expected_stopping}~(b) is that
equality in \eqref{eq:expectedtime} holds if and only if the threshold stopping
algorithm produces exact samples of the max-stable process $Z$, cf.\ 
Section~\ref{sec:thresholdstopping}. Equality in \eqref{eq:expectedtime} in 
this situation was already proved by a different technique in \citet{ord16}.
Naturally, the following respective results for the normalized spectral representations 
(Section~\ref{eq:normalized}) can be recovered as special cases.  

\begin{corollary} \label{cor:expected_stopping}
	\begin{enumerate}[label={(\alph*)},wide=0mm]
		\item {\normalfont (\citeauthor{osz18} \citeyear{osz18})\textbf{.}}
		The expected stopping time of the Threshold Stopping 
		Algorithm~\ref{sec:survey}.\ref{alg:ts} with sup-normalized representation 
		$V=V^{\lVert\cdot\rVert_\infty}$ and threshold 
		$\tau=\theta_{\lVert\cdot\rVert_\infty}$ is 
		\begin{align*}
		\EE (N_{V^{\lVert \cdot \rVert_\infty}}(K)) = \EE \, T_{\theta_{\lVert \cdot \rVert_\infty}} = \theta_{\lVert\cdot\rVert_\infty} \, \EE \, \Big\{ {1}/{\inf_{\bm{x} \in K} Z(\bm{x})} \Big\}.
		\end{align*} 
		\item {\normalfont (\citeauthor{deo16} \citeyear{deo16})\textbf{.}} 
		The expected stopping time of the Threshold Stopping 
		Algorithm~\ref{sec:survey}.\ref{alg:ts} with sum-normalized representation 
		$V=V^{\lVert\cdot\rVert_1}$ and threshold $\tau=N$ is 
		\begin{align*}
		\EE(N_{V^{\lVert \cdot \rVert_1}}(K)) = \EE \, T_N = N \, \EE \, \Big\{ {1}/{\inf_{\bm{x} \in K} Z(\bm{x})} \Big\}.
		\end{align*}
	\end{enumerate}
\end{corollary}

An interesting observation is that the expressions for $\EE(N_V(K)) = \EE \,T_\tau$ above depend
on the law of the spectral process $V$ used only via the law of the resulting
max-stable process $Z$. In particular, if $V$ is any spectral process for $Z$, 
the constant
\begin{align} \label{eq:extremal-coefficient}
\theta_{\lVert\cdot\rVert_\infty} = \EE \, \Big\{ \sup_{\bm{x} \in K} V(\bm{x}) \Big\} 
= -\log \PP\Big\{\sup_{x \in K} Z(x) \leq 1\Big\}
\end{align} 
is usually known as \emph{extremal coefficient} of $Z$ on $K$. For 
$K=\{\bm{x}_1,\ldots,\bm{x}_N\}$ it ranges between $1$ and $N$ and can be 
interpreted as the effective number of independent variables in the set 
$\{Z(\bm x_1),\ldots,Z(\bm x_N)\}$. In view of 
Corollary~\ref{cor:expected_stopping}, being able to effectively simulate from 
a sup-normalized spectral process is therefore a worthwhile endeavor. What is
however unclear in this general setting, is how the computational complexities
$c_{V^{\lVert\cdot\rVert_1}}(K)$ and $c_{V^{\lVert\cdot\rVert_\infty}}(K)$ of 
obtaining a single realization of either $V^{\lVert\cdot\rVert_1}$ or 
$V^{\lVert\cdot\rVert_\infty}$ relate to one another. {Here, a more 
	effective simulation technique for $V^{\lVert\cdot\rVert_1}$ rather than for
	$V^{\lVert\cdot\rVert_\infty}$ may outweigh
        the reduction of the factor 
	$\EE(N_V(K))=\EE(T_\tau)$ by using $V=V^{\lVert\cdot\rVert_\infty}$ instead of 
	$V=V^{\lVert\cdot\rVert_1}$}, see Section~\ref{sec:BRgeneric} for related 
references for the case of Brown-Resnick processes.
This is a general trade-off associated with the choice of the spectral process $V$ for threshold stopping algorithms. Should it be easy to simulate $Z$ from $V$ (low $N_V(K)$) or should be easy to simulate $V$ itself (low $c_V(K)$)? For $V=V^{\lVert\cdot\rVert_1}$ and $V=V^{\lVert\cdot\rVert_\infty}$ we can at least trace $\EE(N_V(K))$ analytically as just discussed.

What is more, \citet{deo16} show that the expected number of simulated spectral
processes in the Extremal Functions Algorithm~\ref{sec:survey}.\ref{alg:ef} 
neither depends on the law of $Z$ nor on the geometry of the domain $K$.

\begin{proposition}[\citealt{deo16}]
	\label{prop:efficiencyEF}
	The expected number of spectral processes to be simulated in the Extremal 
	Functions Algorithm~\ref{sec:survey}.\ref{alg:ef} in order to obtain an exact
	sample of $Z$ on the set $K = \{\bm x_1,\ldots, \bm x_N\}$ equals $N$, i.e.\
	$\EE \, N_V(K) = N$ for this algorithm.
\end{proposition}

Table~\ref{table:efficiency} 
summarizes these findings on the efficiency of the three generic exact 
simulation algorithms considered in this section.  Since the max-stable process
$Z$ to be simulated has standard Fr\'echet margins, we have
\begin{align} \label{eq:infconstantgeqone}
\EE \, \Big\{ {1}/{\inf_{\bm{x} \in K} Z(\bm{x})} \Big\} \geq 1, 
\end{align} 
and equality holds if and only if $Z$ is almost surely constant on $K$. Hence, 
apart from this exceptional case, the expected number of simulated spectral 
processes $\EE(N_V(K))$ for the extremal functions approach is always smaller
than the corresponding number for the sum-normalized approach. According to the 
results in \citet{deo16} (cf.\ also Section~\ref{sec:summary-generic}), the spectral processes involved in the two approaches are very closely related to each other,  i.e.\ their complexities $c_V(K)$ are almost identical. Thus, in terms of computational complexity, the extremal 
functions approach is always preferable to the sum-normalized approach if exact
samples are desired.

\begin{table}
	\centering
	\caption{\small 
		The expected number $\EE(N_V(K))$ of spectral functions $V_j$ to be simulated 
		to obtain an exact sample of a max-stable process $Z$ on a set 
		$K = \{\bm x_1,\ldots, \bm x_N\}$ for three generic simulation algorithms. Each method relies on the ability to simulate from specific spectral functions $V$.
	}
	\label{table:efficiency}

	{\small
\tabcolsep1.5mm
		\begin{tabular}{llcc}
			\toprule
			\multicolumn{2}{l}{\bf Method/Reference} 
			& Spectral fcts.~$V$ 
			& {$\EE(N_V(K))$}
			\\
			\midrule
			\multirow{3}{*}{SN}
			& Sup-normalized threshold stopping 
			& \multirow{3}{*}{$V^{\lVert\cdot\rVert_\infty}$}
			&  \multirow{3}{*}{$\theta_{\lVert\cdot\rVert_\infty} \, \EE \, \big\{ {1}/{\inf_{\bm{x} \in K} Z(\bm{x})} \big\}$} \\
			& (\cite{osz18},\\ & Section~\ref{sec:normalized}) \\[1mm]
			\multirow{3}{*}{DM}
			& Sum-normalized threshold stopping 
			& \multirow{3}{*}{$V^{\lVert\cdot\rVert_1}$}
			&  \multirow{3}{*}{$N \, \EE \, \big\{ {1}/{\inf_{\bm{x} \in K} Z(\bm{x})} \big\}$} \\
			& (\cite{dm15},\\ & Section~\ref{sec:normalized}) \\[1mm]
			\multirow{3}{*}{EF}
			& Extremal functions 
			& \multirow{3}{*}{$V^{(1)},\ldots,V^{(N)}$}
			&  \multirow{3}{*}{$N$} \\
			& (\cite{deo16},\\ & Section~\ref{sec:EF}) \\
			\bottomrule
		\end{tabular}
	}
\end{table}

\begin{remark} \label{rk:lbdm}
	For the record breakers approach, \cite{lbdm16} show that the expected number 
	$\EE(N_V(\{\bm x_1,\ldots, \bm x_N\}))$ of spectral processes 
	$V_j \sim V^{(\text{orig})}$ to be simulated in order to produce an exact 
	sample of a Brown-Resnick process $Z$ lies in $o(N^\eps)$ for any $\eps>0$. The
	result is however difficult to compare with Table~\ref{table:efficiency} as it 
	holds for fixed $K \supset \{\bm x_1,\ldots,\bm x_N\}$ only. For instance, the 
	corresponding result for the sup-normalized threshold stopping could be phrased
	as $\EE(N_V(\{\bm x_1,\ldots, \bm x_N\})) \in \mathcal{O}(1)$ for 
	$V=V^{\lVert\cdot\rVert_\infty}$. This exacerbates meaningful comparisons.
	
\end{remark}

\subsection{Accuracy}

While simulation via normalized spectral functions with appropriate thresholds
or the extremal functions approach produce exact samples from the distribution
of a max-stable process, these algorithms can be computationally expensive. 
Therefore, also the analysis of non-exact simulation algorithms is of interest. 
\vspace*{3mm}

\paragraph{Threshold stopping.}
Our main focus lies on the potentially non-exact Threshold Stopping 
Algorithm~\ref{sec:survey}.\ref{alg:ts} in what follows. As explained in 
Section~\ref{sec:thresholdstopping}, such an algorithm can be non-exact if the 
threshold $\tau$ is exceeded by the spectral process $V$ on $K$ with positive 
probability. Naturally, decreasing the threshold $\tau$ reduces the 
computational cost. But at the same time, the simulation is more likely to be
less accurate as well. To make this specific, let us define the 
\emph{simulation error} as the deviation of the resulting finite approximation
$Z^{(T_{\tau})}$ from the exact sample $Z=Z^{(\infty)}$. The following 
proposition provides a very general description of the distribution of the 
simulation error.

\begin{proposition} \label{prop:general-error}
	For any measurable function $f: C(K) \times K \to [0,\infty]$, we have
	\begin{align*}
	& \PP\big\{ |Z(\bm{x}) - Z^{(T_{\tau})}(\bm{x})| > f(Z^{(T_{\tau})},\bm{x}) \text{ for some } \bm{x} \in K\big\} \\
	={}& 1 - \EE_{Z^{(T_{\tau})}}\bigg\{\exp\bigg(-\EE_V\bigg\{\sup_{\bm{x} \in K} \frac{V(\bm{x})}{Z^{(T)}(\bm{x}) + f(Z^{(T_{\tau})},\bm{x})} 
	- \sup_{\bm{x} \in K} \frac{\tau}{Z^{(T_{\tau})}(\bm{x})} \bigg\}_+\bigg)\bigg\},
	\end{align*}
	where the spectral process $V$ and the stopped process $Z^{(T_{\tau})}$ are 
	stochastically independent.
\end{proposition}

Specifically, by setting $f(Z,\bm{x})=\varepsilon$ or 
$f(Z,\bm{x})= \varepsilon Z(\bm x)$, Proposition \ref{prop:general-error} 
entails the probability that an absolute error of size larger then $\varepsilon$
occurs
\begin{align*}
\notag \mathcal{P}^{\text{\upshape (abs)}}_{\tau,\varepsilon}={}& \PP\Big\{ \sup_{\bm{x} \in K} |Z(\bm{x}) - Z^{(T_{\tau})}(\bm{x})| > \varepsilon \Big\} \\
={}& 1 - \EE_{Z^{(T_{\tau})}}\bigg\{\exp\bigg(-\EE_V\bigg\{\sup_{\bm{x} \in K} \frac{V(\bm{x})}{Z^{(T_{\tau})}(\bm{x}) + \varepsilon} - \sup_{\bm{x} \in K} \frac{\tau}{Z^{(T_{\tau})}(\bm{x})}\bigg\}_+\bigg)\bigg\},
\end{align*}
and the probability that a relative error of size larger than $\varepsilon$ 
occurs 
\begin{align*}
\notag
\mathcal{P}^{\text{\upshape (rel)}}_{\tau,\varepsilon}={}& \PP\bigg\{ \sup_{\bm{x} \in K} \frac{|Z(\bm{x}) - Z^{(T_{\tau})}(\bm{x})|}{Z^{(T_{\tau})}(\bm{x})} > \varepsilon \bigg\} \\
={}& 1 - \EE_{Z^{(T_{\tau})}}\bigg\{\exp\bigg(-\EE_V\bigg\{\sup_{\bm{x} \in K} \frac{V(\bm{x})}{(1+ \varepsilon) Z^{(T_{\tau})}(\bm{x})} - \sup_{\bm{x} \in K} \frac{\tau}{Z^{(T_{\tau})}(\bm{x})}\bigg\}_+\bigg)\bigg\}.
\end{align*}
Both error occurance probabilities are increasing as the error size $\eps$ goes
to zero. For $\eps=0$ they coincide with the probability that a simulation error occurs at all
\begin{align}
\notag \mathcal{P}_{\tau}={}& \PP\Big\{ \sup_{\bm{x} \in K} |Z(\bm{x}) - Z^{(T_{\tau})}(\bm{x})| > 0 \Big\} \\
={}& 1 - \EE_{Z^{(T_{\tau})}}\bigg\{\exp\bigg(-\EE_V\bigg\{\sup_{\bm{x} \in K} \frac{V(\bm{x})}{Z^{(T_{\tau})}(\bm{x})} - \sup_{\bm{x} \in K} \frac{\tau}{Z^{(T_{\tau})}(\bm{x})}\bigg\}_+\bigg)\bigg\},
\label{eq:err}
\end{align}
which can serve as a benchmark.
In the notation of  Section~\ref{sec:EF}, an approximation error occurs 
precisely when the finite approximation $Z^{(T_\tau)}$ does not involve all 
extremal functions $\varphi \in \Phi_+ =\{\varphi_+^{(1)},\ldots,\varphi_+^{(N)}\}$.
The situation gets worse, the larger the number of missing extremal functions
\begin{align*}
M_{\tau} = \bigg\lvert \bigg\{ \Gamma_j^{-1} V_j \in \Phi_+:\ \Gamma_j^{-1} \tau \leq \inf_{\bm{x} \in K} \max_{k < j} \Gamma_k^{-1} V_k(\bm{x})\bigg\} \bigg\rvert.
\end{align*}
is. The expected number of missing extremal functions $\EE(M_{\tau})$ is a 
natural upper bound for the error probability $\mathcal{P}_{\tau}$, i.e.\
$\mathcal{P}_{\tau}\leq \EE(M_{\tau})$.

\begin{proposition} \label{prop:expected_missing} 
	The expected number of missing extremal functions $\EE(M_{\tau})$ in the 
	finite approximation $Z^{(T_\tau)}$ of the max-stable random field $Z$ is bounded by
	\begin{align} \label{eq:new-bound}
	\EE(M_{\tau}) \leq \EE \,  \bigg\{ \sup_{\bm{x} \in K} \frac{V(\bm{x})}{Z(\bm{x})} -  \sup_{\bm{x} \in K} \frac \tau {Z(\bm{x})} \bigg\}_+
	\end{align}
	where the max-stable process $Z$ and the spectral process $V$ are stochastically
	independent.
\end{proposition}

\begin{remark}
	\cite{osz18} showed that \begin{align*}
	\EE \lvert \Phi_{+} \rvert =\EE \, \bigg\{ \sup_{\bm{x} \in K} \frac{V(\bm{x})}{Z(\bm{x})} \bigg\}.
	\end{align*}
	In view of \eqref{eq:expectedtime} and $\EE(\lvert \Phi_+ \rvert) \leq \EE(T_{\tau}) + \EE(M_{\tau})$, we 
	believe that inequality \eqref{eq:new-bound} provides a relatively sharp bound
	for the error term $\EE(M_{\tau})$. In particular, it is sharper than 
	the bound 
	\begin{align*} 
	\EE(M_{\tau}) &
	\leq  
	\EE\,\bigg\{ \bigg(\sup_{\bm{x} \in K} \frac{V(\bm{x})}{Z(\bm{x})}\bigg) \mathbf{1}_{\big\{\sup_{\bm{x} \in K} V(\bm{x}) >\tau\big\}}\bigg\}
	\end{align*}
	in the proof of Proposition 10.4.2 in \cite{ord16}. 
	A significantly simplified (though less sharp) version of \eqref{eq:new-bound}
	is obtained by
	\begin{align*} 
	\EE(M_{\tau}) 
	\leq{}&  \EE \, \bigg\{\sup_{\bm{x} \in K} \frac{V(\bm{x}) - \tau}{Z(\bm{x})} \bigg\}_+ 
	{}\leq{} \EE \, \Big\{ \sup_{\bm{x} \in K} (V(\bm{x}) - \tau)_+\Big\} \,\, \EE \, \Big\{{1}/{\inf_{\bm{x} \in K} Z(\bm{x})}\Big\}.
	\end{align*}
\end{remark}

For both, the error bound  \eqref{eq:new-bound} and the exact error 
\eqref{eq:err}, it is generally difficult to provide analytic expressions. The
precise terms can however be assessed for finite 
$K=\{\bm{x}_1,\ldots,\bm{x}_N\}$ via simulation, see 
Appendix~\ref{app:error-simu} for details. The assessment is based on the 
observations that all the extremal functions of a max-stable process can be 
simulated via the Extremal Functions Algorithm~\ref{sec:survey}.\ref{alg:ef}
and that the potentially relevant non-extremal functions can be simulated 
independently once the extremal functions and the process $Z$ are known. This
allows us to check which of these functions would have been taken into account
by the Threshold Stopping Algorithm~\ref{sec:survey}.\ref{alg:ts}. Hence, we 
can compare the approximation $Z^{(T_{\tau})}$ and the exact realization $Z$
and identify the missing extremal functions. \vspace*{3mm}

\paragraph{Extremal functions.} Besides the threshold stopping algorithm, also 
the Extremal Functions Algorithm~\ref{sec:survey}.\ref{alg:ef} may include a 
simulation error if not all extremal functions $\varphi^{(1)}_+, \ldots,\varphi^{(N)}_+$
are taken into account. Considering the approximation $Z^{(n)}_+$ of 
$Z=Z^{(N)}_+$ on $K=\{\bm{x}_1,\ldots,\bm{x}_N\}$ after the $n$th step of the
extremal functions algorithm as given in \eqref{eq:ef-intermediate} for some 
$n \leq N$ yields the following analogies to Propositions~\ref{prop:general-error}~and~\ref{prop:expected_missing}.

\begin{proposition} \label{prop:sup-err_n}
	For any measurable function $f: C(K) \times K \to [0,\infty]$, we have
	\begin{align*}
	& \PP\big\{ |Z(\bm x) - Z^{(n)}_+(\bm x)| > f(Z^{(n)}_+,\bm x) \text{ for some } \bm x \in K\big\} \\
	={}& 1 - \EE_{Z^{(n)}_+}\bigg\{\exp\bigg(-\EE_V\bigg\{\max_{i=1,\ldots,n} \frac{V(\bm x_i)}{Z^{(n)}_+(\bm x)} - \sup_{\bm x \in K} \frac{V(\bm x)}{Z^{(n)}_+(\bm x) + f(Z^{(n)}_+,\bm x)}\bigg\}_+\bigg)\bigg\},
	\end{align*}
	where the spectral process $V$ and the process $Z^{(n)}_+$ are stochastically independent.
\end{proposition}

\begin{proposition} \label{prop:expected_missing_n}
	The expected number $\EE(M^{(n)}_+)$ of missing extremal functions $M^{(n)}_+$ after $n$ steps of the
	extremal functions algorithm can be computed as 
	\begin{align} \label{eq:error-ef}
	\EE(M^{(n)}_+) ={}& \EE\bigg\{\sup_{\bm x \in K} \frac{V(\bm x)}{Z(\bm x)}\bigg\}
	- \EE\bigg\{\max_{i=1,\ldots,n} \frac{V(\bm x_i)}{Z(\bm x_i)}\bigg\},
	\end{align}
	where the max-stable process $Z$ and the spectral process $V$ are stochastically independent.
\end{proposition}

\section{Comparative numerical study} \label{sec:study}

In order to gain further insights on the comparative performance of the different approaches to max-stable process simulation, specifically as we deviate from the exact setting, the absense of analytic expressions makes it necessary to conduct a broader numerical study. We focus on comparing the three generic and potentially exact methods from Table~\ref{table:efficiency}/Figure~\ref{fig:generic-summary} (DM/EF/SN) when applied to generate (approximate or exact) samples from the widely used classes of 
\begin{enumerate}[label={(\roman*)},itemsep=0mm]
	\item \emph{Brown-Resnick processes} \citep{ksdh09} with spectral representation~\eqref{eq:spec-BR} and fBS variogram $\vario(\bm{h})=2v \lVert \bm{h} \rVert^\alpha$, where $v>0$ and $\alpha \in (0,2)$.
	\item \emph{extremal-$t$ processes} \citep{opitz13} with spectral representation
	\begin{align}
	V(\bm{x})= \frac{\sqrt{\pi} \,2^{1-\nu/2}}{\Gamma((\nu+1)/2)}\, W(\bm{x})_+^\nu, \quad \bm{x} \in K,
	\end{align}
	where $W$ is a standard Gaussian random field with {exponential correlation
		function $\rho(\bm{h}) = \Cov(W(\bm x + \bm h), W(\bm x)) = \exp(-\lVert \bm{h} / s \rVert)$} with scale $s>0$ and the parameter $\nu >0$ influences the degrees of freedom of the underlying multivariate $t$-distribution in the dependence structure, cf.\ \cite{opitz13}.
\end{enumerate}
Both classes of processes are based on underlying Gaussian random fields $W$ and all algorithms considered to obtain an (approximate or exact) sample $Z^{\text{(simulated)}}$ of the associated max-stable process $Z$ are based on repeated sampling from $W$. Therefore, the number $N_W(K)$ of Gaussian processes $W_j \sim W$ needed for one sample of $Z$ constitutes a natural measure for the algorithms' efficiency  and we call the number $\EE(N_W(K))$ that is needed on average the \emph{mean time} in this context. 

\begin{table}
	\centering
	\caption{\small 
		The expected number $\EE(N_W(K))$ of Gaussian processes $W_j$ to be simulated 
		to obtain an exact sample of the associated Brown-Resnick or extremal-$t$ process $Z$ on a set $K = \{\bm x_1,\ldots, \bm x_N\}$ for the three generic simulation algorithms from Table~\ref{table:efficiency}, where the SN algorithm is based on rejection sampling  \citep{defondeville-davison-18}.
The ratio $c_V(K)/c_W(K)$ represents the number of samples from $W$ that are needed to obtain a sample from $V$.
}
	\label{table:efficiencyGaussian}
	{\small
		\begin{tabular}{llccc}
			\toprule
			\multicolumn{2}{l}{\bf Method} 
			& Spectral functions $V$ 
			& {$c_V(K)/c_W(K)$} 
			& Mean time {$\EE(N_W(K))$} \\
			\midrule
			\multirow{2}{*}{SN}
			&
			& \multirow{2}{*}{$V^{\lVert\cdot\rVert_\infty}$}
			& \multirow{2}{*}{$N/\theta_{\lVert\cdot\rVert_\infty}$} 
			&   \multirow{2}{*}{$N \, \EE \, \big\{ {1}/{\inf_{\bm{x} \in K} Z(\bm{x})} \big\}$}\\
			& \\[1mm]
			\multirow{2}{*}{DM}
			&
			& \multirow{2}{*}{$V^{\lVert\cdot\rVert_1}$}
			& \multirow{2}{*}{1} 
			&  \multirow{2}{*}{$N \, \EE \, \big\{ {1}/{\inf_{\bm{x} \in K} Z(\bm{x})} \big\}$}
			\\
			& \\[1mm]
			\multirow{2}{*}{EF}
			&
			& \multirow{2}{*}{$V^{(1)},\ldots,V^{(N)}$}
			& \multirow{2}{*}{1}
			& \multirow{2}{*}{$N$} \\
			& \\
			\bottomrule
		\end{tabular}
	}
\end{table}

For exact simulation, it is possible to derive the precise mean times of each algorithm from the theoretical considerations of Section~\ref{sec:eff-acc}, see Table~\ref{table:efficiencyGaussian}.  While sampling from the spectral processes $V^{(1)}, V^{(2)}, \dots, V^{(N)}$ and $V^{\lVert \cdot \rVert_1}$ is straightforward for Brown-Resnick and extremal-$t$ processes and involves only one Gaussian process simulation for each sample of the respective spectral process  (see \cite{deo16}, \cite{dm15} or Sections~\ref{sec:normalized} and \ref{sec:BRgeneric}), we choose to use the rejection sampling algorithm proposed by \cite{defondeville-davison-18} {based on sum-normalized spectral processes as proposals} to obtain (exact) samples from the sup-normalized spectral function $V^{\lVert\cdot\rVert_\infty}$. In this case we need to take into account the average acceptance rate $\theta_{\lVert\cdot\rVert_\infty}/N$, see Table~\ref{table:efficiencyGaussian}. In view of \eqref{eq:infconstantgeqone} this shows already that for exact simulation, the EF algorithm should be preferred over the DM approach and the SN approach according to the mean time $\EE(N_W(K))$.

\begin{figure}
	{\small
		\centering
		\hspace*{7mm}
		\tabcolsep0.95cm
		\begin{minipage}{12cm}
				\hspace*{1.3mm}
				\tabcolsep8.3mm
			\begin{tabular}{ccc}
				$\bm{\text{\textbf{Variance}} = 0.5}$
				& $\bm{\text{\textbf{Variance}} = 1}$
				& $\bm{\text{\textbf{Variance}} = 2}$
		\end{tabular}\end{minipage}\\
		\rotatebox{90}{
			\begin{minipage}{20cm}
				\hspace*{1.3mm}
				\tabcolsep1.33cm
				\begin{tabular}{ccccc}
					$\bm{\alpha = 1.8}$
					& $\bm{\alpha = 1.4}$
					& $\bm{\alpha = 1.0}$
					& $\bm{\alpha = 0.6}$
					& $\bm{\alpha = 0.2}$
				\end{tabular}
				\vspace*{1mm}
			\end{minipage}
		}
		\includegraphics[width=11.8cm]{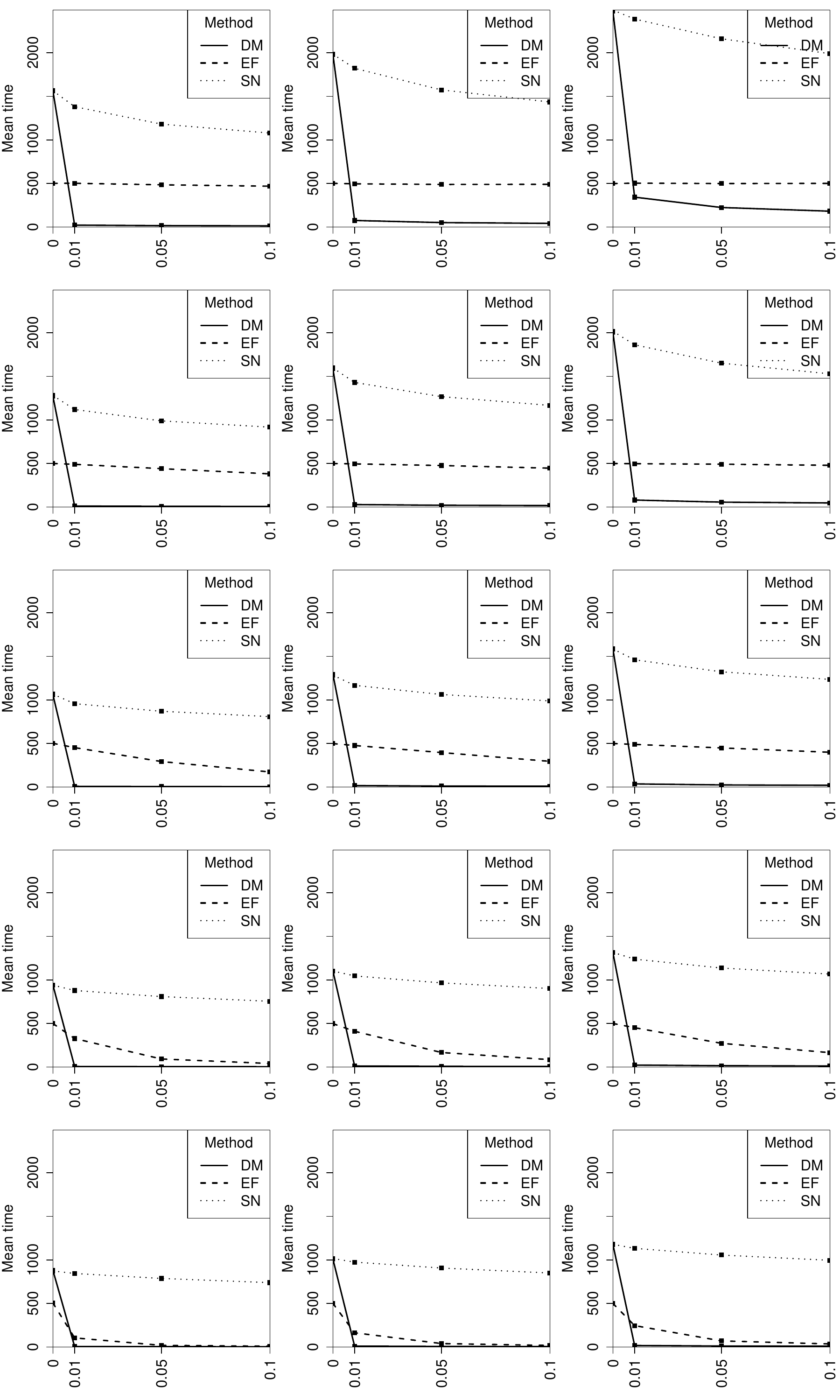}
		\caption{\small Brown-Resnick process simulation in 15 scenarios: Displayed are the mean times of three generic simulation algorithms (DM/EF/SN) for a given tolerated error probability ranging from 0 (``exact simulation'') to 0.1, see Section~\ref{sec:study} and Figure~\ref{fig:BR_DM} for further details.}
		\label{fig:BR4}
	}
\end{figure}

The main purpose of our study is now to investigate the relative efficiency of the algorithms as we vary their accuracy in a reasonable range. As a simulation domain we consider the 501 equi-distantly spaced points $K=\{-1,-0.996,\dots,1\}$ in the interval $[-1,1]$. In the Brown-Resnick case, we consider the parameter scenarios that arise from choosing $\alpha$ in $\{0.2,0.6,1.0,1.4,1.8\}$ (rather noisy to rather smooth) and variance parameter $v$ in $\{0.5,1,2\}$. The extremal-$t$ scenarios consist of $\nu \in \{1,2,4\}$ and scale parameter $s \in \{0.5,1,2\}$. Further we prespecify error probabilities $\mathcal{P} = \PP(Z^{\text{(simulated)}} \neq Z) \in \{0,0.01, 0.05, 0.1\}$ that we are willing to tolerate, whilst observing the corresponding times $N_W(K)$ and estimating the mean time $\EE(N_W(K))$ based on 50\,000 simulations in each case. For the threshold stopping approaches DM and SN the error probability $\mathcal{P}$  coincides  with the benchmark error term $\mathcal{P}_\tau$ in \eqref{eq:err}. That is, for these algorithms we need to select the threshold $\tau$ appropriately in order ensure $\mathcal{P}$ assumes the right value. For the EF approach we deviate from accuracy by fixing an appropriate equi-distantly spaced subset of locations in the simulation domain $K$. The appropriate thresholds and subsets for given error probability $\mathcal{P}$ were also found simulation-based.

\begin{figure}
	\centering
	\includegraphics[width=6cm]{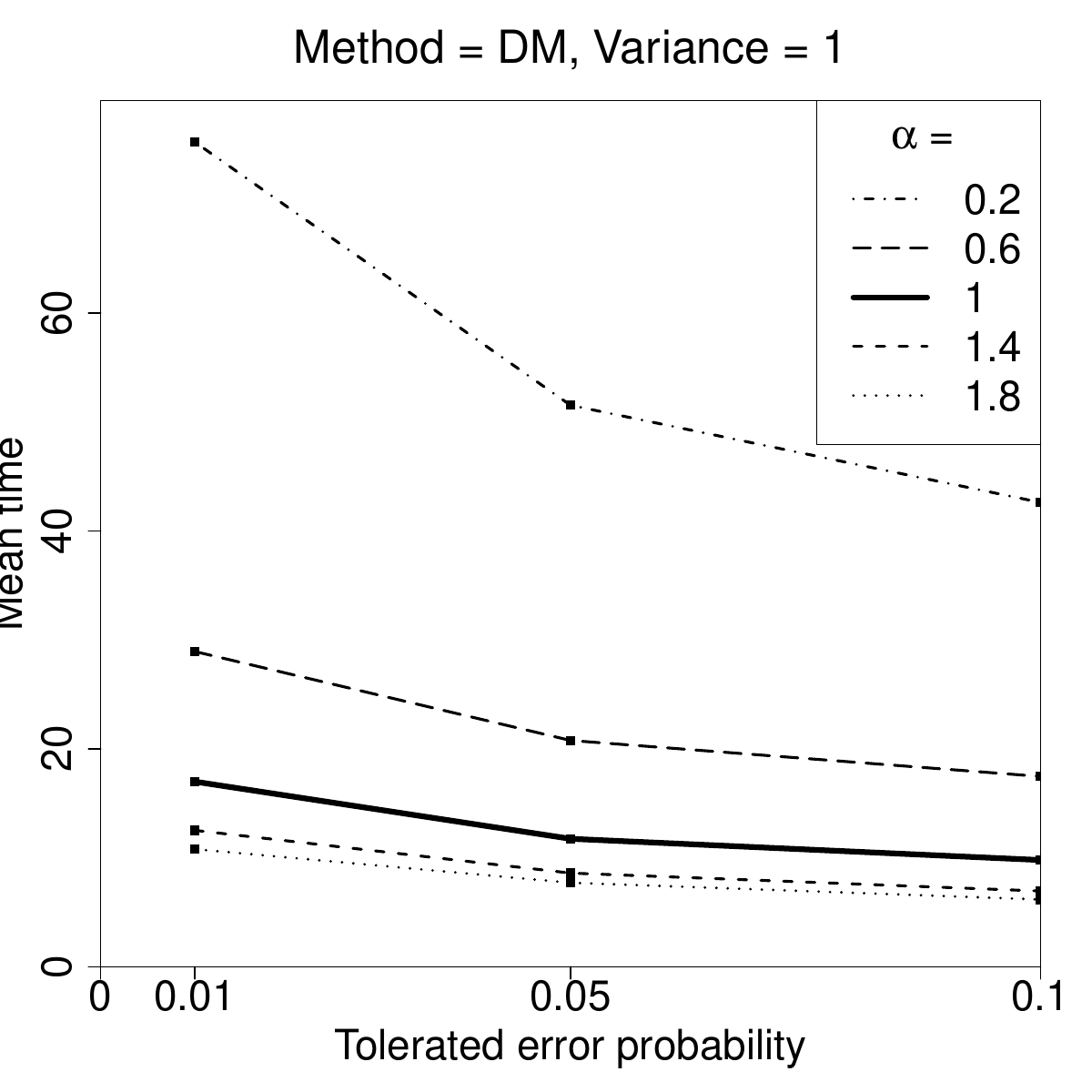}
	\includegraphics[width=6cm]{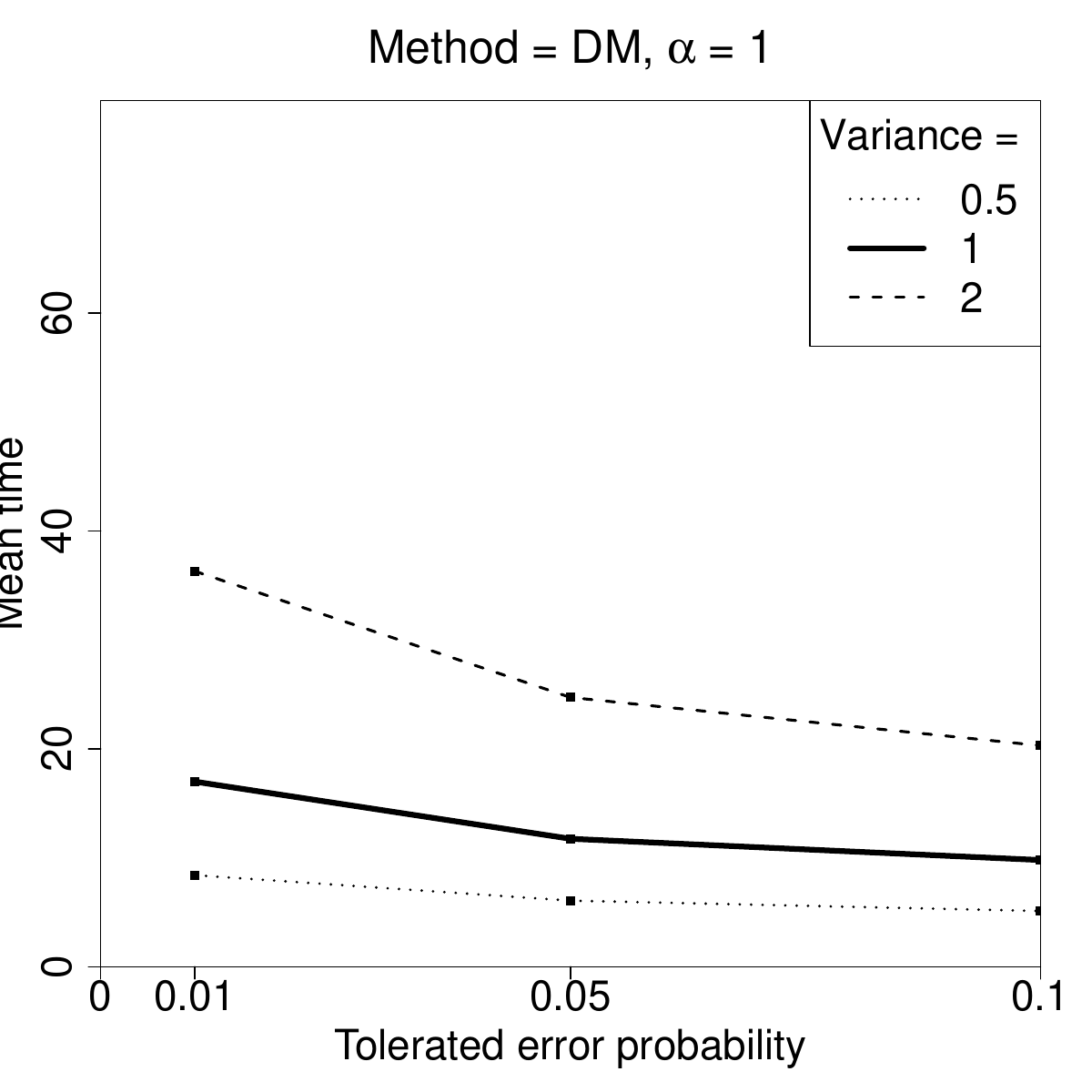}
	\caption{\small A closer look at the mean times of the Dieker-Mikosch (DM) algorithm for Brown-Resnick process simulation. The plots complement Figure~\ref{fig:BR4} revealing the scale and variability of the mean times of the DM algorithm, which may seem reduced to zero in Figure~\ref{fig:BR4}.}
	\label{fig:BR_DM}
\end{figure}

\begin{figure}
	{\small
		\centering
		\hspace*{14mm}
		\tabcolsep1.17cm
		\begin{minipage}{12cm}
                        \tabcolsep1.12cm
			\begin{tabular}{ccc}
				$\bm{\text{\textbf{Scale}} = 0.5}$
				& $\bm{\text{\textbf{Scale}} = 1}$
				& $\bm{\text{\textbf{Scale}} = 2}$
		\end{tabular}\end{minipage}\\
		\rotatebox{90}{
			\begin{minipage}{12cm}
				\hspace*{1mm}
				\tabcolsep1.55cm
				\begin{tabular}{ccccc}
					$\bm{\nu = 4}$
					& $\bm{\nu = 2}$
					& $\bm{\nu = 1}$
				\end{tabular}
				\vspace*{1mm}
			\end{minipage}
		}
		\includegraphics[width=12cm]{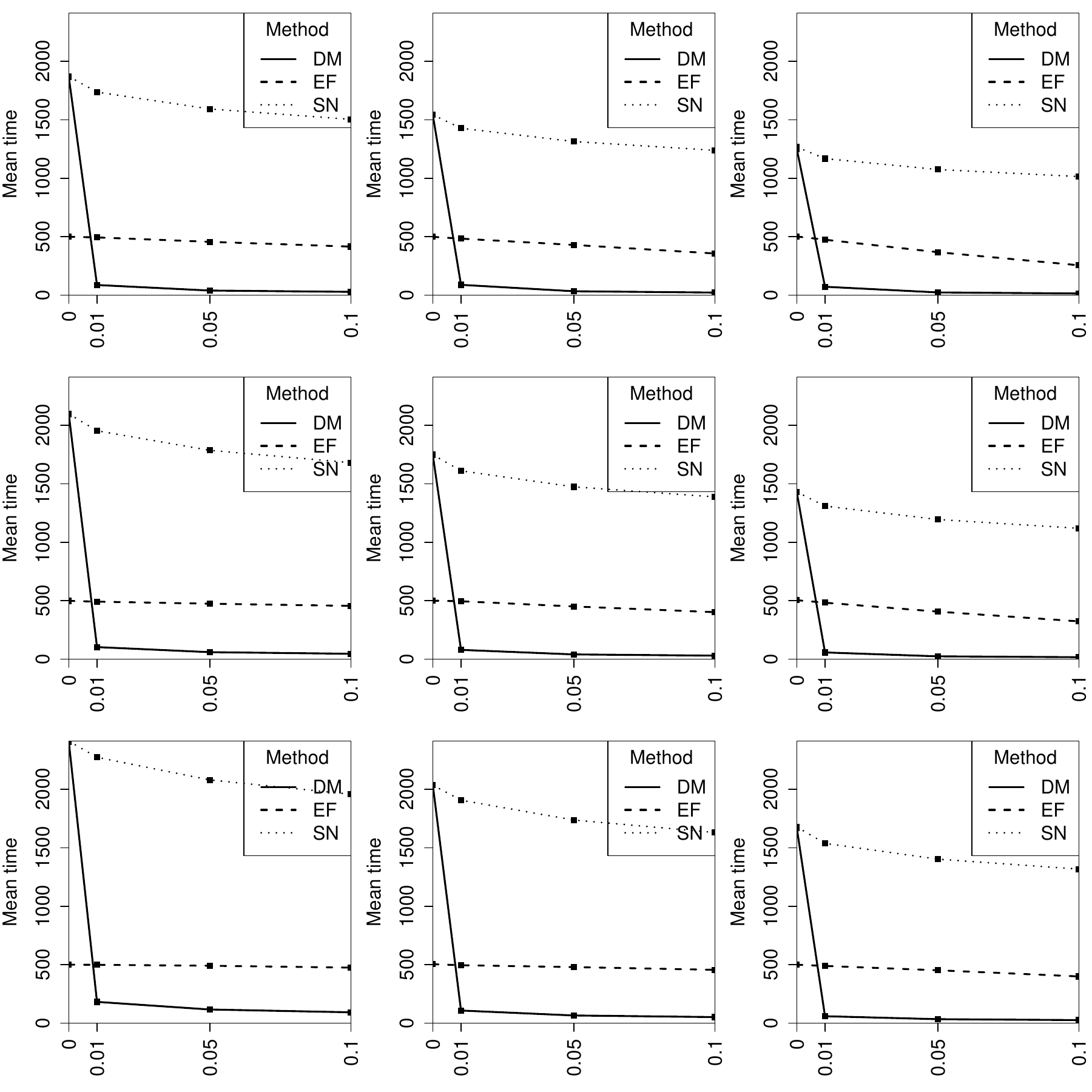}
		\caption{\small Extremal-t-process simulation in 9 different parameter scenarios: Displayed are the mean times of three generic simulation algorithms (DM/EF/SN) for a given tolerated error probability ranging from 0 (``exact simulation'') to 0.1, see Section~\ref{sec:study} for further details.}
		\label{fig:ET4}
	}
\end{figure}

The results of our study for Brown-Resnick processes are reported in Figures~\ref{fig:BR4} and \ref{fig:BR_DM} and for extremal-$t$ processes in Figure~\ref{fig:ET4}. Some of the observations are as expected. The mean time increases in each scenario with lower tolerable error probability. It also increases in the Brown-Resnick case with higher variance and as the processes' roughness increases due to smaller $\alpha$. In the extremal-$t$ case the mean time increases as the scale gets smaller and as the degrees of freedom parameter $\nu$ increases. However, our main interest lies in the relative performance of the three algorithms DM, EF and SN. And whilst for exact simulation ($\mathcal{P} = 0$) the theoretical dominance of the EF approach can be confirmed, the DM approach seems to be uniformly best once we allow for a small tolerable error probability $\mathcal{P} \geq 0.01$. We anticipate a critical value $\mathcal{P}_{\text{critical}}$ closer to zero for the tolerable error probability at which the EF approach will start to perform better than the DM algorithm. The SN approach -- in the form considered here -- cannot match up with either the EF or DM algorithm, chiefly because sampling from the sup-normalized spectral process is costly for Brown-Resnick and extremal-$t$ processes. However, in Section~\ref{sec:discussion} we will point the reader to  modifications of the SN approach and situations, in which it can be very valuable again.

\section{Discussion} \label{sec:discussion}

The simulation of max-stable processes has become an important task as part of spatial risk assessment, specifically in the environmental sciences. The last decade saw several new approaches to the simulation of max-stable processes. The present article provides an overview and compares the generic approaches according to their efficiency in relation to their accuracy. Moving from accurate simulation to tolerating small errors is a major issue of practical concern due to the inherently large computational costs for simulating a max-stable process. We contextualize known theoretical results on the efficiency in the exact setting  (cf.~Tables~\ref{table:efficiency} and \ref{table:efficiencyGaussian}), while adding some new point process based results on the efficiency and accuracy for the approximate setting (cf.~Section~\ref{sec:eff-acc}). An at first sight surprising observation of our numerical study is that the Dieker-Mikosch (DM) approach -- despite being uniformly worse than the extremal functions (EF) approach in the exact setting -- significantly outperforms all generic approaches, once we allow for a small tolerable error probability. 
That said, this finding is  in line with the computational results in \cite{os18} and may be attributed to the DM approach's probabilistic homogeneity of spectral functions. In other words, compared to other algorithms, the DM approach ``converges'' enormously fast to a stochastic process, which is an accurate sample of a max-stable process with very high probability. However, the algorithm fails to be certain and seeks this confirmation for a very long time.

Further, our numerical study might create the impression that the threshold stopping approach using sup-normalized spectral functions (SN) is not worth considering anymore. We would like to correct that impression by emphasizing that the success of this approach depends largely on the ability to efficiently simulate from the sup-normalized spectral process $V^{\lVert\cdot\rVert_\infty}$, which is a research question in its own and of independent interest in other contexts, see also Remark~\ref{rk:pareto}.
In fact, the motivation of \citet{defondeville-davison-18} for introducing the generic rejection sampling approach was not to use it for max-stable process simulation, but to obtain accurate samples from associated Pareto-processes to be readily available for threshold-based inference. Spatio-temporal threshold based inference on extremes is currently an active area of research. 
For Brown-Resnick processes or extremal-$t$ processes sampling from $V^{\lVert\cdot\rVert_\infty}$ is hard and choosing a generic rejection sampling approach for this task is not particularly helpful, which explains the poor performance of the SN approach in our study. 
Alternatives include MCMC approaches \citep{osz18, oss19} and for the class of Brown-Resnick processes modified rejection sampling \citep{oss19} or using the ansatz of \citet{ho-dombry-17}.
For other classes of max-stable processes, the SN approach may well be the most efficient way of exact simulation. For instance, \citet{osz18} show that the sup-normalized process $V^{\lVert\cdot\rVert_\infty}$ can be easily simulated for a broad subclass of mixed moving maxima processes including Gaussian extreme value models \citep{smith90}. Simulation procedures for several of these max-stable models are implemented in the \texttt{R} package \texttt{RandomFields} \citep{RandomFields}.\vspace*{3mm}

We conclude with some practical advice. First of all, we recommend to use exact simulation of max-stable processes, whenever it is feasible. The EF algorithm is designed for this purpose and from our perspective it is the generic approach to use as long as the number of points $N$ in the simulation domain does not get too large. Should exact simulation from the sup-normalized spectral process $V^{\lVert\cdot\rVert_\infty}$ not be too costly, e.\,g.\  for mixed moving maxima processes, then the SN approach can be a worthwhile alternative. In view of the comparison in Table~\ref{table:efficiency} it may even reduce the computational cost significantly, when a large number $N$ of points in a fixed domain $K$ is considered. For approximate simulation, the DM approach seems to perform best -- at least we could not detect a single scenario during our extensive numerical studies in which this was not the case. Unfortunately, the efficiency comes at the price of not knowing when to stop the DM algorithm. We therefore recommend to employ at least additional checks to ensure the obtained sample exhibits reasonable characteristics.  
Alternatively, the SN approach also lends itself as an approximate approach  by means of more efficient MCMC techniques to obtain samples from the sup-normalized spectral process $V^{\lVert\cdot\rVert_\infty}$ as discussed above.

Finally, we would like to mention that for specific classes of max-stable processes, such as Brown-Resnick processes, specific approaches tailored to this class, such as \cite{lbdm16}, may be worth considering, even though meaningful comparisons in terms of efficiency and accuracy seem difficult to achieve, cf.~Remark~\ref{rk:lbdm}, and it is unclear how an approximate version would look like in this case.

\section*{Acknowledgements}
The new theoretical results for this manuscript were obtained during mutual visits of KS at the University of Siegen and MO at Cardiff University. MO and KS thank their hosting institutions for their generous hospitality. The authors are also very grateful for the thoughtful suggestions from the reviewing process. In particular, these comments resulted in the clarification of  the marginal standardization and the inclusion of Section~\ref{sec:dataexample}. This substantial revision was undertaken during summer/autumn 2020. MO thankfully acknowledges financial support by Deutsche Forschungsgemeinschaft (DFG, German Research Foundation) under Germany's Excellence Strategy -- EXC 2075 -- 390740016 at the University of Stuttgart. 

\bibliographystyle{imsart-nameyear}
\bibliography{literature}

\newpage
\appendix

\section{Proofs} \label{app:proofs}

The proofs given in this section rely on the fact that the pairs
$\{(U_i,V_i)\}_{i \in \NN} = \{(\Gamma_i^{-1},V_i)\}_{i \in \NN}$ form a 
Poisson point process $\Pi$ on $(0,\infty) \times C(K)$ with intensity 
$u^{-2} \mathrm{d}u \, \PP\{V \in \mathrm{d}v\}$. We make extensive use of the
Slivnyak-Mecke equation, see e.g.\ {Equation (4.1) in \citet{moeller-03}},
which we recall here for convenience for our situation. To this end, let
$\mathcal{N}$ denote the set of locally finite {simple} counting measures $(0,\infty) \times C(K)$, 
whose $\sigma$-algebra is generated by evaluations on Borel subsets of 
$(0,\infty) \times C(K)$. With a slight (but common and convenient) abuse of
notation by identifying simple counting measures with their induced sets, we
have
\begin{align}\label{eq:slivnyakmecke}
\EE\bigg\{ \sum_{(u,v) \in \Pi} f((u,v),\Pi \setminus \{(u,v)\})\bigg\} 
= \int_{(0,\infty) \times C(K)} \EE_\Pi\big\{ f((u,v),\Pi) \big\} \, u^{-2} \mathrm{d}u \, \PP\{V \in \mathrm{d}v\}.
\end{align}
for any non-negative measurable function $f: ((0,\infty) \times C(K)) \times \mathcal{N} \rightarrow [0,\infty)$.

\begin{proof}[Proof of {\normalfont \bf Proposition \ref{prop:expected_stopping}}]
	By definition of $T_{\tau}$, we have
	\begin{align*}
	\EE(T_{\tau}) &= \EE \, \bigg|\bigg\{(u,\spec) \in \Pi:\ u \tau > \inf_{\bm{x} \in K} \bigvee_{(\tilde u, \tilde \spec) \in \Pi, \ \tilde u > u} \tilde u \tilde \spec(\bm{x})\bigg\}\bigg|\\
	&\geq \EE \, \bigg|\bigg\{(u,\spec) \in \Pi:\   u > \inf_{\bm{x} \in K} \frac{\bigvee_{(\tilde u,\tilde \spec) \in \Pi \setminus \{(u,\spec)\}} \tilde u \tilde \spec(\bm{x})} \tau\bigg\}\bigg|,
	\end{align*}
	with equality {if and only if} $\sup_{\bm x \in K} V(\bm x) \leq \tau$ almost surely.
	Then, the Slivnyak-Mecke equation~\eqref{eq:slivnyakmecke} can be applied to the right-hand side to obtain
	\begin{align*} 
	& \EE \, \bigg|\bigg\{(u,\spec) \in \Pi:\   u > \inf_{\bm{x} \in K} \frac{\bigvee_{(\tilde u,\tilde \spec) \in \Pi \setminus \{(u,\spec)\}} \tilde u \tilde \spec(\bm{x})} \tau\bigg\}\bigg| \\
	&= \int_{C(K)} \int_{C(K)} \int_0^\infty  \mathbf{1}_{\big\{\inf_{\bm{x} \in K} \frac{z(\bm{x})}{\tau} < u\big\}}  \, u^{-2} \, \dd u \, \PP\{V \in \dd \spec\} \, \PP\{Z \in \dd z\} 
	= \EE \bigg\{\sup_{\bm{x} \in K} \frac{\tau} {Z(\bm{x})} \bigg\}	\end{align*}
as desired.
\end{proof}

\begin{proof}[Proof of {\normalfont \bf Proposition \ref{prop:general-error}}]
	Since any condition on $Z^{(T_{\tau})}$ can be rewritten in terms of the restricted 
	point process
	$\Pi ( \,\cdot\, \cap ((\tau^{-1} \inf_{\bm{x} \in K} Z^{(T_{\tau})}(\bm{x}),\infty) \times C(K)))$, 
	we have that, conditional on $Z^{(T_{\tau})}$, the restricted point process 
	
\begin{align*}
\Pi( \,\cdot\, \cap ((0,\tau^{-1} \inf_{\bm{x} \in K} Z^{(T_{\tau})}(\bm{x})) \times C(K)))
\end{align*}
is a Poisson point process with intensity measure $u^{-2} \dd u \, \PP(V \in \dd v)$.
	Consequently,
	\begin{align*}
	& \PP\big\{  \lvert Z(\bm{x}) - Z^{(T_{\tau})}(\bm{x})\rvert > f(Z^{(T_{\tau})},\bm{x}) \text{ for some } \bm{x} \in K \mid Z^{(T)}\big\} \\
	={}& \PP\left\{ \exists (u,v) \in \Pi: u \tau < \inf_{\bm{x} \in K} Z^{(T_{\tau})}(\bm{x}), \right.\\
& \hspace*{2.5cm} \left. \ u v(\bm{x}) > Z^{(T_{\tau})}(\bm{x}) + f(Z^{(T_{\tau})},\bm{x}) \text{ for some } \bm{x} \in K\right\} \displaybreak[0]\\
	={}& 1 - \exp\bigg(- \EE_V\bigg\{\int_{\inf_{\bm{x} \in K} \frac{Z^{(T_{\tau})}(\bm{x}) + f(Z^{(T)},\bm{x})}{V(\bm{x})}}^{\inf_{\bm{x} \in K} \frac{Z^{(T_{\tau})}(\bm{x})}{\tau}} u^{-2} \dd u\bigg\}\bigg) \displaybreak[0]\\
	={}& 1 - \exp\bigg(- \EE_V\bigg\{\sup_{\bm{x} \in K} \frac{V(\bm{x})}{Z^{(T_{\tau})}(\bm{x}) + f(Z^{(T)},\bm{x})} - \sup_{\bm{x} \in K} \frac{\tau}{Z^{(T_{\tau})}(\bm{x})}\bigg\}_+ \,\bigg).
	\end{align*}
	Taking the expectation with respect to $Z^{(T_{\tau})}$ finishes the proof. 
\end{proof}

\begin{proof}[Proof of {\normalfont \bf Proposition \ref{prop:expected_missing}}]
	Let $\Pi_+ = \{ (u,\spec) \in \Pi: \ u \spec(\cdot) \in \Phi_+\}$. Then, we can rewrite
    \begin{align*} 
       M_{\tau} = \bigg\lvert \bigg\{ (u,\spec) \in \Pi_+:\ u \tau \leq \inf_{\bm{x} \in K}
         \bigvee_{(\tilde u, \tilde \spec) \in \Pi, \, \tilde u > u} \tilde u \tilde v(\bm{x})\bigg\} \bigg\rvert
    \end{align*}
    and, hence, 
    \begin{align*}
	&\EE(M_{\tau}) = \EE \, \bigg|\bigg\{(u,\spec) \in \Pi:\ u \spec(\bm{x}) > \bigvee_{(\tilde u,\tilde \spec) \in \Pi \setminus \{(u,\spec)\}} \tilde u \tilde \spec(\bm{x}) \ \text{for some } \bm{x} \in K,  \\
	&    \hspace{4.6cm}                 u \tau \leq \inf_{\bm{x} \in K} \bigvee_{(\tilde u,\tilde \spec) \in \Pi, 
		\, \tilde u > u} \tilde u \tilde \spec(\bm{x})\bigg\}\bigg| \\
	&\leq \EE \, \bigg|\bigg\{(u,\spec) \in \Pi:\  
	\inf_{\bm{x} \in K} \frac{\bigvee_{(\tilde u,\tilde \spec) \in \Pi \setminus \{(u,\spec)\}} \tilde u \tilde \spec(\bm{x})}{\spec(\bm{x})} < u \leq \inf_{\bm{x} \in K} \frac{\bigvee_{(\tilde u,\tilde \spec) \in \Pi \setminus \{(u,\spec)\}} \tilde u \tilde \spec(\bm{x})} \tau\bigg\}\bigg|.
	\end{align*}
	Applying the Slivnyak-Mecke formula~\eqref{eq:slivnyakmecke} gives
	\begin{align*}
\EE(M_{\tau}) 
	&\leq \int_{C(K)} \int_{C(K)} \int_0^\infty  \mathbf{1}_{\big\{\inf_{\bm{x} \in K} \frac{z(\bm{x})}{\spec(\bm{x})} < u \leq \inf_{\bm{x} \in K} \frac{z(\bm{x})}\tau\big\}}  \, u^{-2} \, \dd u \, \PP\{V \in \dd \spec\} \, \PP\{Z \in \dd z\} .
	\end{align*}
	which is equivalent to Inequality~\eqref{eq:new-bound}.
\end{proof}

\begin{proof}[Proof of {\normalfont \bf Proposition~\ref{prop:sup-err_n}}]
	Analogously to the proof of Proposition~\ref{prop:general-error}, we obtain
	\begin{align*}
	& \PP\big\{  |Z(\bm x) - Z^{(n)}(\bm x)| > f(Z^{(n)},\bm x) 
	\text{ for some } \bm x \in K \mid Z^{(n)}\big\} \\
	={}& \PP\left\{ \exists (u,v) \in \Pi: u v(\bm x_i) < Z^{(n)}(\bm x_i), \, i=1,\ldots,n, \right.\\ & \hspace*{2.5cm} \left. u v(\bm x) > Z^{(n)}(\bm x) + f(Z^{(n)},\bm x) \text{ for some } \bm x \in K\right\}\\
	={}& 1 - \exp\bigg(- \EE_V\bigg\{\int_{\inf_{\bm x \in K} \frac{Z^{(n)}(\bm x) + f(Z^{(n)},\bm x)}{V(\bm x)}}^{\min_{i=1,\ldots,n} \frac{Z^{(n)}(\bm x_i)}{V(\bm x_i)}} u^{-2} \mathrm{d}u\bigg\}\bigg)\\
	={}& 1 - \exp\bigg(-\EE_V\bigg\{\max_{i=1,\ldots,n} \frac{V(\bm x_i)}{Z^{(n)}(\bm x_i)} - \sup_{\bm x \in K} \frac{V(\bm x)}{Z^{(n)}(\bm x) + f(Z^{(n)},\bm x)}\bigg\}_+\bigg).
	\end{align*}
	Taking the expectation with respect to $Z^{(n)}$ finishes the proof. 
\end{proof}

\begin{proof}[Proof of {\normalfont \bf Proposition~\ref{prop:expected_missing_n}}]
	The expected number of missing extremal functions after the $n$th step
	can be expressed as
	\begin{align}
	\EE(M^{(n)}_+) ={}& \EE \, \bigg|\bigg\{(u,\spec) \in \Pi_+:\ u \spec(\bm x_i) <  \bigvee_{(\tilde u, \tilde \spec) \in \Pi} \tilde u \tilde \spec(x_i) \text{ for all } i=1,\ldots,n\bigg\}\bigg| 
	\end{align}
	Consequently, we have that $\EE(M^{(n)}_+)$ equals
	\begin{align*}
	 \EE \, \bigg|\bigg\{(u,\spec) \in \Pi:\  \inf_{\bm x \in K} \frac{\bigvee_{(\tilde u,\tilde \spec) \in \Pi \setminus \{(u,\spec)\}} \tilde u \tilde \spec(\bm x)}{\spec(\bm x)} < u 
	< \min_{i=1}^n \frac{\bigvee_{(\tilde u,\tilde \spec) \in \Pi \setminus \{(u,\spec)\}} 
		\tilde u \tilde \spec(\bm x_i)}{\spec(\bm x_i)}\bigg\}\bigg|.
	\end{align*}
	The Slivnyak-Mecke equation~\eqref{eq:slivnyakmecke} can be applied to the right-hand side. Hence, we obtain that $\EE(M^{(n)}_+)$ equals
	\begin{align*}
	\int_{C(K)} \int_{C(K)} \int_0^\infty  \mathbf{1}_{\big\{\inf_{\bm x \in K} \frac{z(\bm x)}{\spec(\bm x)}
		< u \leq \min_{i=1,\ldots,n} \frac{z(\bm x_i)}{\spec(\bm x_i)} \big\}} \, u^{-2} \, \dd u \, \PP\{V \in \dd \spec\}
	\, \PP\{Z \in \dd z\} .
	\end{align*}
	The latter coincides with \eqref{eq:error-ef}.
\end{proof}

\section{Simulation assessment of the accuracy} \label{app:error-simu}

This section is a step-by-step description how one can assess the approximation error of the Threshold Stopping Algorithm~\ref{sec:survey}.\ref{alg:ts} with threshold $\tau$ by simulation.

\begin{enumerate}[label={(\arabic*)}]
\itemsep2mm
	\item Use the Extremal Functions Algorithm~\ref{sec:survey}.\ref{alg:ef}        to simulate $\Phi_+$. Denote the resulting functions by $\varphi_1,\ldots,\varphi_k$ and set 
	$Z(\bm x) = \max_{i=1,\ldots,k} \varphi_i(\bm x)$, $\bm x \in K$.\smallskip\\
	Note that only the products of the type $\varphi = U V$ are obtained this way, the 
	components $U$ and $V$  with $(U,V) \in \Pi_+$ (as in the proof of Proposition
	\ref{prop:expected_missing}) are unknown.       

\item Simulate the set $\Pi_+ = \{(U_i^+, V_i^+)\}_{i=1,\ldots,k}$
	conditional on $U_i^+ \cdot V_i^+ = \varphi_i$ for $i=1,\ldots,k$.\smallskip\\
Here, the $U_i^+$ are independent with density
	\begin{align*} f_{U_i^+}(u) \propto u^{-2} \PP( V \in \mathrm{d} \varphi_i /u) \end{align*}
	and $V_i^+ = (U_i^+)^{-1} \varphi_i$.

	\item Simulate the entire set $\Pi_{\min{},-} = \{(U_i^-,V_i^-)\}_{i=1,\ldots,l}$ 
	of all non-extremal functions with 
	$U_i^{-} > U_{\min{}} := \min_{i=1,\ldots,k} U_i^+$. \smallskip\\
These form a Poisson point process with intensity 
	$u^{-2} \mathrm{d} u \cdot \PP(V \in \mathrm{d}v)$
	restricted to the set 
	$\{ (u,v) \in [U_{\min{}},\infty) \times C(K): \ u \cdot v(\bm x) < Z(\bm x) \text{ for all } \bm x \in K\}$.

	\item Merge $\Pi_+$ and $\Pi_{\min{},-}$ to the set $\Pi_{\min{}} = \{(U_i,V_i)\}_{i=1,\ldots,k+l}$
	labelling the points in such a way that $U_i > U_{i+1}$.
	\item Set $Z^{(j)}(\bm x) = \max_{i=1,\ldots,j} U_i \cdot V_i(\bm x)$ and define
	\begin{align*}
	T := \min\bigg\{j \in \{1,\ldots,k+l-1\}: \ U_{j+1} \tau < \inf_{\bm x \in K} Z^{(j)}(\bm x) \bigg\},
	\end{align*}
i.e.\ $T$ equals the stopping time $T_{\tau}$ provided that $T_{\tau} < k+l$.
	Otherwise, we have $T=\infty$.

	\item If $T = \infty$, the stopping criterion does not apply before all the 
	extremal functions are simulated, i.e.\ $Z^{(T)} = Z$ and there is no
	error. Otherwise, i.e.\ if $T < k+l$, there is an error. This can either
	be measured in terms of the absolute/relative deviation between $Z$ and
	$Z^{(T)}$ or in terms of the number of missing extremal functions, i.e.\ the cardinality of 
	the set 
	$\{i \in \{1,\ldots,k\}: \ U_i^+ < U_T \}$.
\end{enumerate}
Repeating this procedure, the average error is an
unbiased estimator of the expected simulation error of the Threshold Stopping Algorithm~\ref{sec:survey}.\ref{alg:ts}.

\end{document}